\documentclass[aps,preprint,singlecolumn,superscriptaddress,nofootinbib,floatfix]{revtex4-2}
\usepackage{amsmath,amssymb,amsthm,bm}
\usepackage{booktabs,array,graphicx,rotating}
\usepackage{adjustbox}
\newlength{\simulationtableheight}
\usepackage{microtype}
\usepackage{graphicx}
\usepackage{xcolor}
\usepackage[colorlinks=true,linkcolor=blue!50!black,citecolor=blue!50!black,urlcolor=blue!50!black]{hyperref}

\usepackage[a4paper,margin=2.4cm]{geometry}

\newcommand{\Tr}{\mathrm{Tr}}
\newcommand{\dd}{\mathrm{d}}
\newcommand{\hc}{^{\dagger}}
\newcommand{\lcr}{\lambda_{\mathrm{cr}}}
\newcommand{\ket}[1]{|#1\rangle}
\newcommand{\bra}[1]{\langle #1|}
\newcommand{\calA}{\mathcal{A}}
\newcommand{\calD}{\mathcal{D}}
\newcommand{\Dt}{D_t}

\newcommand{\n}{\mathbf n}
\newcommand{\J}{\mathbf J}
\newcommand{\bq}{\mathbf q}
\newcommand{\lamcr}{\boldsymbol\lambda_{\rm cr}}

\newcommand{\kE}[2]{|{#1},{#2}\rangle_{\n}}
\newcommand{\PE}[2]{P_{{#1},{#2}}}

\newcommand{\topstrut}{\rule{0pt}{2.9ex}}
\newcommand{\botstrut}{\rule[-1.4ex]{0pt}{0pt}}
\renewcommand{\theequation}{\arabic{equation}}
\renewcommand{\theHequation}{main.\arabic{equation}}
\newtheorem{theorem}{Theorem}
\newtheorem{lemma}{Lemma}
\newtheorem{corollary}{Corollary}
\theoremstyle{remark}
\newtheorem{remark}{Remark}
\begin{document}

\title{\Large Quantum Hamiltonian time crystals from covariant  \\  quantum Landau--Lifshitz--Gilbert relaxation}

\author{Antti J. Niemi}
\email{Antti.Niemi@su.se}
\affiliation{Nordita, Stockholm University,
Hannes Alfv\'ens v\"ag 12, SE-106 91 Stockholm, Sweden}
\affiliation{Wilczek Quantum Center,  Shanghai Institute for Advanced Studies, 
			University of Science and Technology of China, Shanghai 201315, China}


\begin{abstract}
Classical Hamiltonian time crystals are constrained energy minima that move
along symmetry orbits at frequencies set by Lagrange multipliers. We construct
quantum counterparts by minimizing the energy over density matrices with fixed
spectrum and charge expectations. A symmetry-covariant quantum
Landau--Lifshitz--Gilbert equation selects them as periodic attractors,
preserving the spectrum and fixing the multipliers dynamically. A chiral spin
trimer motivated by Cu$_3$ molecular magnets realizes this motion at a
prepared magnetization, without a Zeeman term.
\end{abstract}

\maketitle

\paragraph*{Introduction.}
The original time-crystal proposal sought persistent periodic motion in the lowest available energy state of a physical system
\cite{Wilczek2012,ShapereWilczek2012}. Most subsequent work has concerned driven or open systems. These include
discrete time crystals under periodic driving and continuous time crystals  under dissipation, both realized experimentally
\cite{Else2016,Khemani2016,Yao2017,Zhang2017,Choi2017,Iemini2018,Kongkhambut2022,Sacha2018,Zaletel2023}.
For closed Hamiltonian systems in equilibrium, no-go theorems exclude spontaneously rotating ground states and Gibbs states
\cite{Bruno2013,WatanabeOshikawa2015}.

In classical Hamiltonian systems with continuous symmetries, a different
possibility arises when energy minimization is restricted to fixed values
of conserved charges $Q_a$  \cite{DNPW2019,ADN2020}. Such a minimum is a 
relative equilibrium~\cite{MarsdenRatiu},   it is  not a critical 
point of $H$ but of $H+\lambda^aQ_a$, with
$\lambda^a$ the Lagrange multipliers. Hamilton's equation transports it along
the symmetry orbit, with frequencies given by the multipliers. When the
symmetry acts nontrivially on the minimum, this is a classical Hamiltonian
time crystal. Examples include Coleman's $Q$-balls
\cite{Coleman1985,LeePang1992}, closed molecular rings \cite{DNPW2019},
and condensate vortices~\cite{GDN2021}.

Here we carry this construction into quantum theory on a Hilbert
space $\mathcal H$. We seek a density matrix $\rho_*$ that
minimizes $\langle H\rangle=\operatorname{Tr}(\rho H)$ among
all density matrices with a prescribed spectrum and fixed
expectation values $\langle Q_a\rangle=q_a$ of conserved charges,
$[H,Q_a]=0$.  The factorization
$\rho=\widehat{\Psi}\widehat{\Psi}^\dagger$ into Uhlmann
amplitudes~\cite{Uhlmann1986,Uhlmann1991} supplies canonical
Hamiltonian variables, with the corresponding Schr\"odinger
evolution inducing the von Neumann equation for
$\rho$~\cite{Tronci2019}.
If $[H,\rho_*]\neq0$, the minimizing density matrix evolves
nontrivially under the Hamiltonian
extending the pure-state
construction of~\cite{Niemi2025}.

To identify $\rho_*$ we introduce a symmetry-covariant generalization of the 
quantum Landau--Lifshitz--Gilbert (qLLG) equation \cite{Liu2024,LandauLifshitz1935,Gilbert2004}.
Its isospectral flow decreases the energy but preserves the charge expectations and
fixes  their  Lagrange multipliers dynamically.
The relaxation selects $\rho_*$  as a time-periodic
attractor, with frequencies determined by the multipliers.

We illustrate the mechanism with a quantum trimer  of three spins
coupled by their scalar chirality. The molecular motivation is
provided by triangular $\mathrm{Cu}_3$ complexes, whose three Cu(II)
ions possess localized spin-$\tfrac12$ moments \cite{Trif2008}.
Their low-energy spin-$\tfrac12$ manifold carries a chirality
degree of freedom, proposed for encoding quantum information
in triangular spin clusters \cite{GeorgeotMila2010}.

%
%
%
%
%
%


\vskip 0.2cm

\paragraph*{Hamiltonian von Neumann equation.}
Following Uhlmann \cite{Uhlmann1986,Uhlmann1991} we introduce a $d\times r$ complex matrix $\Psi$
with $d= {\rm dim} \ \mathcal H$ and $r \ge\mathrm{rank}\,\rho$. The  decomposition $\rho=\Psi \Psi\hc $
ensures that $\rho$ is Hermitian and positive. It is also  invariant under 
a gauge transformation that sends $\Psi\to \widehat{\Psi} = \Psi V$, where $V\in U(r) $ is a $d\times r$ unitary matrix \cite{HJW1993}.  
For a time-dependent $V(t)$ define $\calA\equiv i\hbar\,\partial_t VV\hc$, $\Dt\widehat 
\Psi \equiv\partial_t {\widehat \Psi }-\tfrac{i}{\hbar}\widehat \Psi \calA$ 
and for  $G \equiv \Psi \hc \Psi $ we set $\widehat G\equiv \widehat \Psi\hc\widehat \Psi =VGV\hc$. 
With $\nu$ a Lagrange multiplier we introduce the first-order action 
\begin{eqnarray}
S=\!\int\!\dd t\Bigl\{\tfrac{i\hbar}{2}\Tr\bigl(\widehat \Psi \hc\Dt\widehat \Psi -(\Dt\widehat \Psi )\hc\widehat 
\Psi \bigr)\qquad\qquad&\nonumber\\
-\Tr(\widehat \Psi \hc H\widehat \Psi )+\nu\bigl(\Tr\widehat G-1\bigr)\Bigr\}.&
\label{eq:action}
\end{eqnarray}
The kinetic term equips the Uhlmann amplitudes with the Poisson bracket 
$\{\widehat \Psi_{i\alpha}, {\widehat \Psi}^\star_{j\beta}\}=\delta_{ij}\delta_{\alpha\beta}/i\hbar$.
Variation with respect to $\widehat{\Psi}^\dagger$
gives the Schr\"odinger-type equation
\begin{equation}
i\hbar\,\Dt\widehat{\Psi}
=(H-\nu)\widehat{\Psi}.
\label{eq:eom1}
\end{equation}
This induces the von Neumann equation for
$\rho=\widehat{\Psi}\widehat{\Psi}^\dagger$,
independently of $\nu$ and $\calA$, with 
variation of  $\nu$ imposing the
normalization,
\begin{equation}
i\hbar\dot\rho=[H,\rho],
\qquad \Tr\rho=1.
\label{eq:NN}
\end{equation}
Since $H-\nu$ acts on $\widehat{\Psi}$ from the
left, while the $U(r)$ connection $\calA$ acts
from the right, Eq.~(\ref{eq:eom1}) also implies the Gauss law
\begin{equation}
\Dt\widehat G\equiv
\dot{\widehat G}
-\frac{1}{i\hbar}[\calA,\widehat G]=0,
\label{eq:eom}
\end{equation}
It  preserves the eigenvalues
of $\widehat G$, whose nonzero eigenvalues coincide
with those of the density matrix $\rho$.
The solution 
\[
\rho(t)=e^{-iHt/\hbar}\rho_0e^{iHt/\hbar}
\]
of ~\eqref{eq:NN} therefore evolves isospectrally on the
unitary orbit $\mathcal O_{\rho_0}=\{\mathcal U\rho_0\mathcal U\hc 
\} $ of the initial state $\rho_0$, 
where $\mathcal U $ ranges over all $d\times d$ unitary matrices.

Finally, for any fixed operators  $A$ and $B$ on $\mathcal H$  with expectation values $\langle 
A\rangle=\Tr(\rho A)$, $\langle B\rangle=\Tr(\rho B)$, 
the Poisson bracket of Uhlmann amplitudes gives the Lie--Poisson bracket \cite{MarsdenRatiu,Kibble1979,GKM2005}
\begin{equation}
\{\langle A\rangle,\langle B\rangle\}=\frac{1}{i\hbar}\langle[A,B]\rangle, 
\label{eq:LP}
\end{equation}
and  \eqref{eq:NN} becomes the canonical Hamilton's equation with Hamiltonian $\langle H\rangle$.

%
%
%
%
%
%
%
%
%
%
%
%


\vskip 0.2cm

\paragraph*{The quantum  time crystal.}
The expectation values $\langle H\rangle$ and $\langle Q_a\rangle$ of the Hamiltonian and its conserved charges 
are $U(r)$-gauge-invariant functions on the canonical phase space of the Uhlmann amplitudes.
Since the von Neumann time evolution is isospectral,  following Ref.~\cite{ADN2020} we seek a quantum time crystal as a 
time-dependent minimum of $\langle H\rangle$ on an isospectral orbit of density matrices, subject to fixed charge
expectations $\langle Q_a\rangle=q_a$.
%
%
With $P=\mathrm{diag}(p_1,\dots,p_r)$ the eigenvalues of an initial $\rho_0$ and with $G=\Psi \hc \Psi=P$ imposed,  the  Lagrange
multiplier theorem \cite{MarsdenRatiu}  implies that the quantum time crystal is
a  stationary point of
\begin{align}
F=\Tr(\Psi\hc H\Psi)+\lambda^a\bigl[\Tr(\Psi\hc Q_a\Psi)-q_a\bigr]\qquad&\nonumber\\
-\Tr\bigl(\Lambda(\Psi\hc \Psi-P)\bigr),&
\label{eq:F}
\end{align}
where $\lambda^a$ and the Hermitian $r\times d$ matrix $\Lambda$ are Lagrange multipliers.
Stationarity gives $(H+\lambda^aQ_a)\Psi=\Psi\Lambda$ so that $(H+\lambda^aQ_a)\rho=\Psi\Lambda \Psi \hc=\rho\,
(H+\lambda^aQ_a)$. With $\lambda^a_{cr}$
the values of the multipliers at the minimizer $\rho_*$ its
von Neumann time evolution is
\begin{eqnarray} 
i \hbar \dot \rho_* &=& [ H , \rho_*] = - \lambda_{cr}^a [ Q_a , \rho_*] \label{eq:crit1} \\
\Rightarrow \ \ \rho_*(t) &=& e^{\frac{i}{\hbar}\lcr^aQ_at}\rho_*e^{-\frac{i}{\hbar}\lcr^aQ_at},
\label{eq:crit}
\end{eqnarray}
which is  time dependent whenever $\lcr^a[Q_a,\rho_*]\neq0$. The critical values $\lcr^a$ are time independent and
unique provided  the $[Q_a,\rho_*]$ are linearly 
independent. They are  given by $\lcr^a=-\partial E_{\min}/\partial q_a$, where $E_{\min}(q)$ is the constrained minimum energy,
i.e., the value of $F$ in \eqref{eq:F} evaluated at the minimizer, with  the constraint terms vanishing.
Being neither a ground state nor a Gibbs state of the generator $H$ of time evolution, 
this quantum time crystal evades the no-go theorems \cite{Bruno2013,WatanabeOshikawa2015}.

%
%
%
%
%
%
%
%
%
%
%


\vskip 0.2cm

\paragraph*{Covariant quantum LLG dynamics.}
The quantum Landau--Lifshitz--Gilbert equation (qLLG) is  \cite{Liu2024}  
\begin{equation}
i\hbar\dot\rho= [H+\hbar\kappa\dot\rho,\rho],
\label{eq:qLLG}
\end{equation}
with  $\kappa>0$ a dimensionless damping, and for $\kappa=0$ we recover the von Neumann equation.
This equation describes a dissipative  flow of the  density matrix towards a 
time-independent fixed point, generically an energy minimum,  with $\langle H\rangle$ as a Lyapunov function,
\[
\frac{ \mathrm d}{\mathrm dt} \langle H\rangle =-\kappa\hbar\Tr(\dot\rho^2)\le 0 .
\]
Since $H+\hbar\kappa\dot\rho$  is Hermitian the flow   is isospectral for any $\kappa$.

We seek a generalization of (\ref{eq:qLLG}) that instead of a fixed point describes dissipation 
towards a time crystal (\ref{eq:crit1}), (\ref{eq:crit}) as an attracting time periodic
trajectory.  To separate the symmetry rotation from the relaxation,
we introduce
\begin{equation}
\calD_t\rho\equiv
\dot\rho+\frac{1}{i\hbar}\lambda^a[Q_a,\rho].
\label{eq:D}
\end{equation}
where $\lambda^a(t)Q_a$ generates a rotating symmetry frame  We then generalize \eqref{eq:qLLG} to
\begin{eqnarray}
i\hbar\dot\rho &=&
[H+\hbar\kappa\calD_t\rho,\rho],
\label{eq:cov1}\\
\Leftrightarrow\quad
i\hbar\calD_t\rho &=&
[H+\lambda^aQ_a+\hbar\kappa\calD_t\rho,\rho].
\label{eq:cov}
\end{eqnarray}
The first equation has the form of a von Neumann equation with the Hermitian generator
$H+\hbar\kappa\calD_t\rho$, ensuring that the flow
preserves the spectrum and normalization of $\rho$.
The second equation expresses the same dynamics in a form which is manifestly covariant
relative to the symmetry rotation  with $H+\lambda^aQ_a$ as the Hamiltonian.

The equations (\ref{eq:cov1}), (\ref{eq:cov}) provide us with  a relaxation
scheme for constrained energy minimization: We start from a density matrix $\rho(0)$ with prescribed spectrum
and charge expectations $\langle Q_a\rangle=q_a$.  From (\ref{eq:cov1}) we get
\begin{equation}
\frac{\dd\langle Q_a\rangle}{\dd t}
=i\kappa\,\Tr\bigl(\calD_t\rho\,[Q_a,\rho]\bigr).
\label{eq:charge}
\end{equation}
We enforce charge conservation by choosing the multipliers $\lambda^a(t)$ so that the right-hand
side vanishes at every time: At fixed $\rho$, Eq.~(\ref{eq:cov}) determines $\calD_t\rho$
affinely in the multipliers, so this requirement gives a linear system for $\lambda^a$.
The solution is unique whenever the symmetry directions $[Q_a,\rho]$ are linearly independent.
With this choice, the flow remains on the constraint surface $\langle Q_a\rangle=q_a$
%
%
%
while the energy decreases according to
\begin{equation}
\frac{\dd\langle H\rangle}{\dd t}
=-\kappa\hbar\,\Tr\bigl[(\calD_t\rho)^2\bigr]\le0 ,
\label{eq:descent}
\end{equation}
This  dissipative evolution can be used to find constrained
energy minima $\rho_*$ such that
\begin{equation} 
\calD_t\rho_* =0 \ \Rightarrow \ 
[H+\lcr^aQ_a,\rho_*]=0  
\label{eq:minimum}
\end{equation}
which is the stationarity condition of the constrained functional \eqref{eq:F} that underlies Eq.~\eqref{eq:crit1}.  With 
$\lcr^a[Q_a,\rho_*]\neq0$,  the density matrix $\rho_*(t)$  evolves along the
symmetry orbit  according to Eq.~(\ref{eq:crit}).  We note that the time dependence  is different from 
the dynamical phase of an energy eigenstate:  The eigenvectors of $\rho_*$ superpose different charge sectors, so
$H$ rotates them by the symmetry, and the relative phases are observable.

%
%
%

%
%
%
%


\vskip 0.2cm

\paragraph*{Application: A quantum trimer.}
Triangular $\mathrm{Cu}_3$ molecular magnets, which contain three localized
electronic spin-$\frac12$ moments \cite{Trif2008}, 
have been proposed as hosts for chirality qubits
\cite{GeorgeotMila2010}. Motivated by this we consider a minimal model
of three spins coupled through their scalar chirality,
\begin{equation}
H=\tfrac{\hbar\omega}{2}\,\boldsymbol\sigma_1\!\cdot(\boldsymbol\sigma_2\times\boldsymbol\sigma_3),
\qquad
J_a=\tfrac12\sum_i\sigma_i^a,
\label{eq:triad}
\end{equation}
with $[H,J_a]=0$ and $[J_a,J_b]=i\epsilon_{abc}J_c$. 
The Hilbert space decomposes as $\tfrac12\otimes\tfrac12\otimes\tfrac12 =\tfrac32\oplus\tfrac12\oplus\tfrac12$. 
The Hamiltonian vanishes on the total-spin-$\frac32$ quadruplet $|\tfrac32,m\rangle$ with $m$ the eigenvalues 
of $J_3$ and takes the values
$\pm\varepsilon=\pm\sqrt3\,\hbar\omega$ on the two total-spin-$\frac12$ 
doublets $|\pm,m\rangle$ ($m=\pm\tfrac12$). Setting  $\omega>0$ selects the negative-chirality
doublet as the ground manifold. Besides the $SU(2)$ symmetry, $H$ is 
invariant under cyclic permutations of
the sites and odd under transpositions and time reversal. Up to scale it is
the unique $SU(2)$-invariant Hermitian three-spin operator that is odd under
time reversal. 
Microscopically, $H$ is a measure of  the orbital current circulating around the
triangle. 
As an effective theory $H$ can be generated by a magnetic flux through the ring at third
order in the Hubbard hopping \cite{SenChitra1995,BBMK2008}, by
Dzyaloshinskii--Moriya coupling combined with spin polarization
\cite{Trif2008}, or, with helicity-controlled sign, by circularly polarized
driving \cite{Claassen2017}.  

The spin magnetic moment $\boldsymbol\mu=-g\mu_B\langle\mathbf J\rangle$ of the trimer 
vanishes in the Gibbs state $\rho_G=e^{-\beta H}/Z$, with $\beta$ the inverse temperature.
It is distinct from the orbital moment associated with the current measured by $H$, which is a static
$SU(2)$-invariant quantity. A nonzero $\boldsymbol\mu$ must therefore be
prepared, and we prescribe $\langle\mathbf J\rangle=q\,\mathbf n$ with
$\mathbf n$ a unit vector and $q\ge0$. In Cu$_3$ spin-anisotropic interactions 
leave only the component normal to the plane of the triangle conserved~\cite{Trif2008},
and therefore we fix $\mathbf n$ along the normal. For  the shifted constraints
$G_a\equiv\langle J_a\rangle-q\,n_a$ the Lie--Poisson bracket~\eqref{eq:LP}
gives
\begin{equation}
\{G_a,G_b\}=\tfrac1\hbar\,\epsilon_{abc}\,(G_c+q\,n_c),
\qquad
\{\langle H\rangle,G_a\}=0.
\label{eq:Galg}
\end{equation}
For $q\neq0$ only $\mathbf n \cdot \mathbf G$ is first class and generates
rotations about $\mathbf n$,  the two transverse
components of $\mathbf G$ form a second-class pair rotating the direction of $\mathbf n$. 
This rotation leaves  the spectrum of $\rho$  and the energy unchanged but reorients  the Lagrange multiplier
$\boldsymbol\lambda_{\rm cr}=- [\partial E_{\min}/\partial  q] (q)\,\mathbf n $.

With the  initial state $\rho_0$ on the unitary orbit of the Gibbs state $\rho_G$, the flow
(\ref{eq:cov1}), (\ref{eq:cov}) describes states  with the thermal spectrum and fixed von Neumann entropy. 
The eigenvalues of $\rho_G$ are $p_\pm=e^{\mp\beta\varepsilon}/Z$
on the $|\pm\rangle$ doublets and $p_0=1/Z$ on the quadruplet, with
$Z=8\cosh^2(\beta\varepsilon/2)$ and
$E_G=\langle H\rangle_G=-\varepsilon\tanh(\beta\varepsilon/2)$. 

We first minimize $\langle H\rangle$ directly on the orbit subject to $G_a=0$, and 
identify the pair $(E_{\min},\rho_*)$. For this we  saturate von Neumann's trace
inequality applied to $H+\lambda\,\mathbf n\cdot\mathbf J$.
With details given in the Supplementary Material  \cite{SM} we find for the  minimum energy 
\begin{equation}
E_{\min}(q)=\varepsilon\Big[\tfrac{q}{2}-\tanh\tfrac{\beta\varepsilon}{2}\Big],
\qquad 0\le q\le\tanh\tfrac{\beta\varepsilon}{2}.
\label{eq:Emin2}
\end{equation}
and for the  minimizer  we find
\begin{equation}
\begin{aligned}
\rho_* =  \rho_G
& +(p_--p_0)\big(|a\rangle\langle a|-|{-},-\tfrac12\rangle\langle{-},-\tfrac12|\big) \\
& +  (p_0-p_+)\big(|b\rangle\langle b|-|\tfrac32,-\tfrac32\rangle\langle\tfrac32,-\tfrac32|\big)
\end{aligned}
\label{eq:family1}
\end{equation}
\begin{equation}
\begin{aligned}
& \!\!\!  |a\rangle = 
\cos\theta\left|-, -\tfrac12\right\rangle
+ e^{i\phi}\sin\theta
\left|\tfrac32, \tfrac32\right\rangle,
\ \ \   0\leq\theta\leq\frac{\pi}{2},
\\[2mm]
& \!\!\! |b\rangle = 
\cos\vartheta \left|\tfrac32, -\tfrac32\right\rangle
+ e^{i\varphi}\sin\vartheta 
\left|+, \tfrac12\right\rangle,
\ \ \  0\leq\vartheta\leq\frac{\pi}{2}.
\end{aligned}
\label{eq:family2}
\end{equation}
with ($\phi,\varphi$) arbitrary and with the prescribed
magnetization imposing
\begin{equation}
2(p_--p_0)\sin^2\theta
+
2(p_0-p_+)\sin^2\vartheta=q.
\label{eq:angle_constraint}
\end{equation}
%

For $0<q<\tanh(\beta\varepsilon/2)$ and generic ($\theta,\vartheta$) 
the commutators $[J_a,\rho_*]$ are linearly
independent. Consequently, the Lagrange multipliers are uniquely
determined as $\boldsymbol\lambda_{\rm cr}
=-\frac{\varepsilon}{2}\mathbf n$, independently of both
$\beta$ and $q$. According to~(\ref{eq:crit}), $\rho_*$
rotates about $\mathbf n$ with angular velocity
$|\boldsymbol\lambda_{\rm cr}|/\hbar
=\varepsilon/(2\hbar)=\sqrt3\,\omega/2$.
Since the two components of each superposition differ by
$\Delta m=2$, the Bloch-sphere azimuths evolve as
$\phi(t)=\phi-\varepsilon t/\hbar$ and
$\varphi(t)=\varphi-\varepsilon t/\hbar$, while the
latitudes remain constant. The two Bloch vectors thus
undergo common periodic precession with period
$T=2\pi/(\sqrt3\,\omega)$, with the latitudes and the
time-independent relative phase $\phi-\varphi$ determined by 
the initial state.

Single-spin operators cannot connect
the two components of $|a\rangle$, or of $|b\rangle$, which differ by
$\Delta m=2$. Consequently, the expectation values
$\langle\boldsymbol\sigma_i\rangle=\tfrac23 q\,\mathbf n$ are time
independent, while, for generic members of the family, 
two-spin correlations such as $\langle\sigma_i^+ \sigma_j^+\rangle>,$i $i\not=j$, oscillate in
time with the same period as $\rho_*$. 
The density matrix~(\ref{eq:family1})
therefore realizes a Hamiltonian quantum time crystal
without a Zeeman term in the Hamiltonian.

Note that a prescribed spin magnetization does not itself imply
time-dependent motion: magnetized states commuting with
$H$ remain stationary. The moving constrained minima
instead satisfy
$[H,\rho_*]=-\lambda_{\rm cr}[J_{\mathbf n},\rho_*]\neq0$.
Their motion is carried by coherences between different
charge and energy sectors, not by the presence of
a circulating current. Dephasing these coherences preserves
the energy and charge expectations but changes the
density-matrix spectrum, and therefore leaves the
isospectral constraint surface.

The two end points of \eqref{eq:Emin2} are static. At $q=0$ the minimizer
is the Gibbs state itself: $\rho_G$ minimizes $\langle H\rangle$ on its
orbit without any constraint and, having equal weights within each
multiplet of $H$, it commutes with $\mathbf J$ and already has
$\langle\mathbf J\rangle=0$. Its stationarity is in line with the no-go
theorems \cite{Bruno2013,WatanabeOshikawa2015}. At
$q=\tanh(\beta\varepsilon/2)$ both pairs are fully exchanged and $\rho_*$
is again static.

In $\mathrm{Cu}_3$,  the dominant  antiferromagnetic  Heisenberg exchange, absent from (\ref{eq:triad}) 
and, commuting with $H$ and $\mathbf J$,  shifts the spin-$\frac32$ quadruplet up by 
exchange gap $\Delta$ relative to the
centre of the two chirality doublets. For $\Delta>5\varepsilon$ the cheapest transfer through the quadruplet costs
$(\Delta-\varepsilon)/2$ per unit of $q$, more than the $2\varepsilon$ of the
transfer between the doublets, so the minimum mixes the doublets through
$$
|c\rangle=\cos\chi\,|{-},-\tfrac12\rangle+e^{i\eta}\sin\chi\,|{+},\tfrac12\rangle
, \ \ \  0\le\chi\le\tfrac{\pi}{2}
$$
The density matrix
\begin{equation}
\rho_*=\rho_G+(p_--p_+)\big(|c\rangle\langle c|-|{-},-\tfrac12\rangle\langle{-},-\tfrac12|\big)
\label{rho}
\end{equation}
carries $\langle\mathbf J\rangle=q\,\mathbf n$ with $q=(p_--p_+)\sin^2\chi$. It
is the constrained minimum for $0\le q\le p_--p_+$, with
$E_{\min}=E_G+2\varepsilon q$ and
$\boldsymbol\lambda_{\rm cr}=-2\varepsilon\,\mathbf n$, where $E_G$ and the
Gibbs weights are those of the exchange-shifted spectrum. Its coherence
between the two chirality doublets is
\begin{equation}
\langle{-},-\tfrac12|\rho_*(t)|{+},\tfrac12\rangle
=\sqrt{q(p_--p_+-q)}\,e^{i(2\varepsilon t/\hbar-\eta)}.
\label{coherence}
\end{equation}
This amplitude vanishes at both endpoints, where the density matrix
is stationary. For intermediate values of $q$,
the density matrix is a time crystal that
undergoes periodic motion with period
$T=\pi\hbar/\varepsilon$.

%
%
%
%

%
%
%
\vskip 0.2cm
\emph{Numerical verification.}  We proceed to verify these exact  trimer results 
as numerical solutions of (\ref{eq:cov1}), (\ref{eq:cov}).  We integrate  the $8\times8$ Uhlmann amplitudes in the gauge
$\mathcal A=0$, with $\nu=0$ and $\widehat{\Psi}(0)^\dagger \widehat{\Psi}(0)=P$, the prescribed spectrum.
In the instantaneous eigenbasis of $\rho$, Eq.~(\ref{eq:cov}) determines
$\mathcal D_t\rho$ algebraically and affinely in $\lambda^a$ \cite{SM}
so that \eqref{eq:charge} is a real $3\times3$ linear system for the multipliers,
solved at every stage of an adaptive eighth-order Runge--Kutta integration
with relative tolerance $10^{-10}$, at $\kappa=0.5$. Initial amplitudes are
Gaussian random matrices with orthonormalized columns scaled by $P^{1/2}$,
followed by a least-squares correction to
$\langle\mathbf J\rangle=q\,\mathbf n$ that preserves $\widehat{\Psi}^\dagger \widehat{\Psi}=P$; no
information about the minimizers enters. The spectrum, hence the entropy,
and $\langle\mathbf J\rangle$ stay constant to $10^{-11}$ along every run;
all runs are tabulated in \cite{SM}.

On Gibbs orbits, for $\beta\varepsilon=\tfrac12,1,2,4$ and
$0\le q\le0.97\tanh(\beta\varepsilon/2)$, every run converges to the family
\eqref{eq:family1}--\eqref{eq:angle_constraint}: the populations outside the two pairs of \eqref{eq:family2} and all
coherences other than the two pair coherences return to their Gibbs values,
the angles read off the populations of $|\tfrac32,\tfrac32\rangle$ and
$|{+},\tfrac12\rangle$ satisfy \eqref{eq:angle_constraint}, and the energies agree with \eqref{eq:Emin2} to
$10^{-9}\varepsilon$ [Fig.~1(a)]. Runs with the same $q$ end at different
$(\theta,\vartheta)$, as the degeneracy of the family implies, and the fully
exchanged end points lie on the line $E=-\varepsilon q/2$ for every $\beta$,
since $E_G=-\varepsilon\,q_{\max}$. For $q>0$ the multipliers converge to
$\boldsymbol\lambda_{\rm cr}=-\tfrac{\varepsilon}{2}\mathbf n$ to relative
accuracy $10^{-6}$, both azimuths advance at the rate $-\varepsilon/\hbar$ with
constant difference $\phi-\varphi$,
$\langle\boldsymbol\sigma_i\rangle=\tfrac{2q}{3}\mathbf n$ on all sites, and
$\langle\sigma_1^+\sigma_2^+\rangle$ oscillates with period $2\pi/\sqrt3\,\omega$
[Fig.~1(c), top]. At $q=0$ the runs reach $\rho_G$ itself,
$\|\rho-\rho_G\|<10^{-6}$, while the multiplier, which \eqref{eq:charge} no longer fixes
there, settles on the sphere $|\boldsymbol\lambda|=\varepsilon/2$ in a direction
set by the initial state. The covariant velocity decays exponentially in all
cases, at rates $0.03$--$0.05\,\omega$ for $\beta\varepsilon=1$ and more slowly
at lower temperature, whereas the physical velocity $\|\dot\rho\|$ stays
finite for $q>0$ and vanishes only at $q=0$ [Fig.~1(b)]. 

With the exchange shift ($\Delta=6\varepsilon$, $\beta\varepsilon=1$,
$0\le q\le0.8\,(p_--p_+)$) every sampled run ends on the doublet-mixing
branch (\ref{rho}): energy, multiplier, mixing angle and coherence agree
with $E_G+2\varepsilon q$, $-2\varepsilon\,\mathbf n$,
$\sin^2\chi=q/(p_--p_+)$ and (\ref{coherence}) to $10^{-9}$, the local
moments precess as derived above, with the transverse amplitude,
$2\pi/3$ phase lags, vanishing transverse sum and period
$\pi\hbar/\varepsilon$ reproduced to the same precision
[Fig.~1(c), bottom].  At $q=0$ the sampled runs again reach $\rho_G$, along paths on which
$|\boldsymbol\lambda|\to2\varepsilon$.



%
%
%
%

\begin{figure*}[t]
\includegraphics[width=\textwidth]{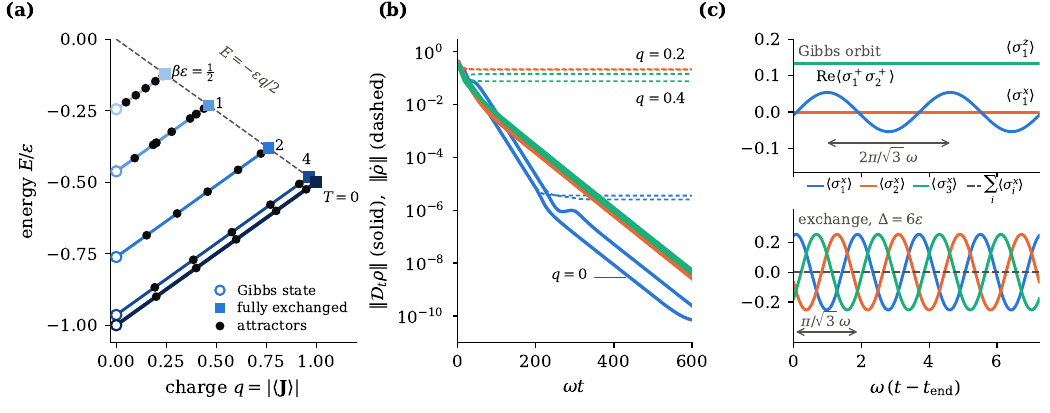}
\caption{Covariant relaxation of the quantum trimer, $\kappa=0.5$,
$\varepsilon=\sqrt3\,\hbar\omega$. (a)~Lines: $E_{\min}(q)$, Eq.~\eqref{eq:Emin2}, for
$\beta\varepsilon=\tfrac12,1,2,4$ and $T=0$; dots: energies of the attractors
reached from random initial states; circles and squares: the static end
points $q=0$ and $q=\tanh(\beta\varepsilon/2)$; dashed: $E=-\varepsilon q/2$.
(b)~Covariant velocity $\|\mathcal D_t\rho\|$ (solid) and physical velocity
$\|\dot\rho\|$ (dashed) for $\beta\varepsilon=1$ and $q=0,\,0.2,\,0.4$, two
runs each. (c)~Observables after relaxation at $\beta\varepsilon=1$, $q=0.2$.
Top, Gibbs orbit: $\langle\sigma_1^z\rangle$, $\langle\sigma_1^x\rangle$ and
$\mathrm{Re}\langle\sigma_1^+\sigma_2^+\rangle$; the arrow marks the period
$2\pi/\sqrt3\,\omega$. Bottom, exchange-dominated orbit, $\Delta=6\varepsilon$:
$\langle\sigma_i^x\rangle$ for the three sites and their sum (dashed); the
arrow marks the period $\pi/\sqrt3\,\omega$.}\label{fig:gibbs}
\end{figure*}



%
%
%
\paragraph*{Discussion.}
We have defined Hamiltonian quantum time crystals as time-dependent density
matrices that minimize the energy at fixed spectrum and fixed expectation
values of the conserved charges. We have constructed them as periodic attractors
of a symmetry-covariant quantum Landau--Lifshitz--Gilbert equation. Its
classical counterpart is the standard tool for magnetization
relaxation~\cite{LandauLifshitz1935,Gilbert2004}, and for a single
spin-$\tfrac12$ Eq.~\eqref{eq:qLLG} reduces to it with Gilbert damping
$\kappa$~\cite{Liu2024}. Equations~\eqref{eq:cov1}--\eqref{eq:cov} extend it to interacting spins,
with the spectrum of $\rho$ in place of the magnetization length and the
symmetry rotation exempted from damping, so that generic initial states flow
onto uniformly rotating orbits. These relative
equilibria~\cite{MarsdenRatiu} are neither ground states nor Gibbs states and
do not contradict the no-go theorems~\cite{Bruno2013,WatanabeOshikawa2015}.
As in the classical case, the relaxation is phenomenological rather than
thermalization, and its microscopic origin remains open.

Our example is a chiral spin trimer. It is motivated by antiferromagnetic
$\mathrm{Cu}_3$ molecular magnets~\cite{Trif2008,GeorgeotMila2010} and
quantizes the classical triangular time crystal of Ref.~\cite{DNPW2019}.
A prepared magnetization makes it rotate without a Zeeman term. With
chirality coupling alone, the single-spin expectation values stay constant
and only two-spin correlations oscillate. With dominant exchange, the local
spins themselves precess, while the total magnetization stays fixed. The
chirality is not needed for the motion itself: a Heisenberg trimer on its
Gibbs orbit also has a moving minimum, with the frequency set by the
exchange gap. The chirality selects the frequency and the observable that
carries it.

The construction applies to any spin triangle that conserves its total
spin, or one component of it. Candidates are other molecular spin
triangles~\cite{Trif2010}, such as $\{\mathrm{V}_{15}\}$ and
$\{\mathrm{Fe}_3\}$~\cite{Chiorescu2000,Boudalis2018}, triangular triple
quantum dots~\cite{HsiehHawrylak2010}, and three-qubit loops with synthetic
flux~\cite{Roushan2017,Lienhard2020}, where site-resolved measurements in suitable spin bases could
resolve the oscillating two-spin correlations. Coupled trimers could test
whether the temporal order becomes collective.

\section*{Acknowledgements}
This research is supported by the Swedish Research Council under Contract No. 2022-04037 
and by COST Action CA23134 (POLYTOPO).  

The simulation codes  were written with
the assistance of large language models (both Anthropic Claude and OpenAI models) under the
author's direction and are available  in a single repository at \url{https://github.com/niemiantti/qLLG} 
 All runs were performed by the author on local hardware, and all code, analyses and figures
were checked by the author, who takes full responsibility for them. AI tools were  also used in support of
theoretical analyses and for language  editing of the manuscript.

\clearpage

\setcounter{equation}{0}
\setcounter{section}{0}
\setcounter{table}{0}
\renewcommand{\theequation}{S\arabic{equation}}
\renewcommand{\thesection}{S\arabic{section}}
\renewcommand{\thetable}{S\arabic{table}}
\renewcommand{\theHequation}{supp.\arabic{equation}}
\renewcommand{\theHsection}{supp.\arabic{section}}
\renewcommand{\theHtable}{supp.\arabic{table}}

\begin{center}
{\Large Supplemental Material }
\end{center}

\noindent
This Supplement gives complete derivations for the spin-trimer section of
the Letter: the Poisson brackets of the Uhlmann amplitudes, the algebra of
the shifted constraints and why the Lagrange multiplier is parallel to the
charge (Sec.~\ref{sec:constraints}); von Neumann's trace inequality with its
equality conditions (Sec.~\ref{sec:trace}); the Lagrange-multiplier bound
and the selection of the multiplier $\lambda_{\rm cr}=-\varepsilon/2$ by a
level crossing (Sec.~\ref{sec:bound}); the complete set of minimizers, the
charge window and the result
$E_{\min}(q)=\varepsilon[q/2-\tanh(\beta\varepsilon/2)]$
(Sec.~\ref{sec:minimizers}); the motion of the minimizing family and its
observable signatures (Sec.~\ref{sec:motion}); the limiting cases
(Sec.~\ref{sec:limits}); the exchange-dominated branch with its precessing
local moments (Sec.~\ref{sec:exchange}); and the numerical method with all
runs of the Letter tabulated (Sec.~\ref{sec:numerics}). Throughout,
$\varepsilon=\sqrt3\,\hbar\omega$, and in all numerical values
$\hbar=\omega=1$ and $\kappa=0.5$. Unless a limit is stated explicitly we
assume $\omega>0$ and $0<\beta<\infty$, so that the Gibbs weights obey
$p_->p_0>p_+$; the limits $\beta\to\infty$ and $\beta\to0$ are treated in
Sec.~\ref{sec:limits}. The $J_a$ are dimensionless, $\hbar J_a$ being the
physical angular momentum, and entropies are measured in units of $k_B$.

\section{Model, spectrum and joint eigenbasis}
\label{sec:model}

The trimer Hamiltonian \eqref{eq:triad} of the Letter and its conserved charges are
\begin{equation}
H=\tfrac{\hbar\omega}{2}\,\boldsymbol\sigma_1\!\cdot(\boldsymbol\sigma_2\times\boldsymbol\sigma_3),
\qquad J_a=\tfrac12\sum_{i=1}^3\sigma_i^a,\qquad
[H,J_a]=0,\qquad [J_a,J_b]=i\epsilon_{abc}J_c .
\label{supp:eq:model}
\end{equation}
The eight-dimensional Hilbert space decomposes into a total-spin-$\frac32$
quadruplet and two total-spin-$\frac12$ doublets. The scalar chirality
$\boldsymbol\sigma_1\!\cdot(\boldsymbol\sigma_2\times\boldsymbol\sigma_3)$
vanishes on the quadruplet and takes the values $\mp2\sqrt3$ on the two
doublets, so that
\begin{equation}
E_-=-\varepsilon,\qquad E_Q=0,\qquad E_+=+\varepsilon,\qquad
\varepsilon=\sqrt3\,\hbar\omega ,
\end{equation}
on the negative-chirality doublet $\ket{-}$, the quadruplet and the
positive-chirality doublet $\ket{+}$, respectively.

For a unit vector $\n$ we use the joint eigenbasis $\kE{\mu}{m}$ of $H$ and
$\n\!\cdot\!\J$, which exists because the two operators commute. The first
label is a multiplet label, $\mu\in\{-,\tfrac32,+\}$: the two chirality
doublets are distinct copies of total spin $\tfrac12$ and are distinguished
by their chirality $\mu=\pm$, while $\mu=\tfrac32$ is the quadruplet. The
eight states, with their energies and magnetic quantum numbers along $\n$,
are
\begin{equation}
\begin{array}{c|cccccccc}
 & \kE{-}{+\tfrac12} & \kE{-}{-\tfrac12} & \kE{\tfrac32}{\tfrac32} & \kE{\tfrac32}{\tfrac12}
 & \kE{\tfrac32}{-\tfrac12} & \kE{\tfrac32}{-\tfrac32} & \kE{+}{+\tfrac12} & \kE{+}{-\tfrac12}\botstrut\\ \hline
\topstrut E & -\varepsilon & -\varepsilon & 0 & 0 & 0 & 0 & +\varepsilon & +\varepsilon\\
m & +\tfrac12 & -\tfrac12 & +\tfrac32 & +\tfrac12 & -\tfrac12 & -\tfrac32 & +\tfrac12 & -\tfrac12
\end{array}
\label{supp:eq:levels}
\end{equation}
Explicitly, with $\n=\hat{\mathbf z}$ and $\omega_3=e^{2\pi i/3}$,
\begin{equation}
\kE{\tfrac32}{\tfrac32}=\ket{\uparrow\uparrow\uparrow},\qquad
\kE{-}{-\tfrac12}=\tfrac{1}{\sqrt3}\big(\ket{\downarrow\downarrow\uparrow}
+\omega_3\ket{\downarrow\uparrow\downarrow}+\omega_3^2\ket{\uparrow\downarrow\downarrow}\big),
\label{supp:eq:explicit}
\end{equation}
the latter having negative chirality for the orientation of \eqref{supp:eq:model}. The
remaining states follow by the ladder operators with the Condon--Shortley
phases, $\kE{\mu}{m+1}\propto J_+\kE{\mu}{m}$ with a positive coefficient,
and the positive-chirality doublet by the exchange
$\omega_3\leftrightarrow\omega_3^2$ of the site phases. These conventions
fix the phases of the observables in Secs.~\ref{sec:motion} and
\ref{sec:exchange}.

The cyclic permutation $C_3$ of the sites commutes with $H$ and with $\J$.
The quadruplet states are symmetric under $C_3$ and the two chiral doublets
are eigenstates of $C_3$ with eigenvalues $\omega_3^{\pm1}$. This gives a
lemma used repeatedly below.

\begin{lemma}[single-spin expectations]
\label{lem:single}
In every joint eigenstate, $\bra{\mu,m}\boldsymbol\sigma_i\ket{\mu,m}_{\n}=\tfrac{2m}{3}\,\n$
for each site $i=1,2,3$.
\end{lemma}
\begin{proof}
The components of $\boldsymbol\sigma_i$ transverse to $\n$ change $m$ by
$\pm1$ and have no diagonal elements. For the component along $\n$,
$C_3^\dagger\,\sigma_i^{\n}\,C_3=\sigma_{i+1}^{\n}$ and $C_3\kE{\mu}{m}$ is
$\kE{\mu}{m}$ up to a phase, so the diagonal expectation is the same on all
three sites. Their sum is $\bra{\mu,m}2\,\n\!\cdot\!\J\ket{\mu,m}=2m$.
\end{proof}

\section{Poisson brackets, shifted constraints and the direction of the multiplier}
\label{sec:constraints}

\paragraph{Poisson bracket of the Uhlmann amplitudes.}
The kinetic term of the action \eqref{eq:action} of the Letter,
$\frac{i\hbar}{2}\Tr(\widehat{\Psi}^\dagger  D_t \widehat{\Psi} -(D_t\widehat{\Psi})^\dagger\widehat{\Psi})$, is first
order in time derivatives and makes the matrix elements of $\widehat{\Psi}$ and of
$\widehat{\Psi}^\star$ canonically conjugate,
\begin{equation}
\{\widehat{\Psi}_{i\alpha},\widehat{\Psi}^\star_{j\beta}\}=\frac{1}{i\hbar}\,\delta_{ij}\delta_{\alpha\beta},
\qquad
\{\widehat{\Psi}_{i\alpha}, \widehat{\Psi}_{j\beta}\}=\{\widehat{\Psi}^\star_{i\alpha},\widehat{\Psi}^\star_{j\beta}\}=0 ,
\label{supp:eq:PBW}
\end{equation}
with $i,j=1,\dots,d$ the Hilbert-space index and $\alpha,\beta=1,\dots,r$
the gauge index. For a fixed operator $A$ the expectation value
$\langle A\rangle=\Tr(\rho A)=\sum_{i,j,\alpha}\widehat{\Psi}^\star_{j\alpha}A_{ji}\widehat{\Psi}_{i\alpha}$
is a quadratic function of the amplitudes, and the chain rule gives for two
such functions
\begin{equation}
\begin{aligned}
\{\langle A\rangle,\langle B\rangle\}
&=\frac{1}{i\hbar}\sum_{i,\alpha}\Big[\frac{\partial\langle A\rangle}{\partial\widehat{\Psi}_{i\alpha}}
\frac{\partial\langle B\rangle}{\partial\widehat{\Psi}^\star_{i\alpha}}
-\frac{\partial\langle B\rangle}{\partial\widehat{\Psi}_{i\alpha}}
\frac{\partial\langle A\rangle}{\partial\widehat{\Psi}^\star_{i\alpha}}\Big]\\
&=\frac{1}{i\hbar}\Big[\Tr(\widehat{\Psi}^\dagger AB\widehat{\Psi})-\Tr(\widehat{\Psi}^\dagger BA\widehat{\Psi})\Big]
=\frac{1}{i\hbar}\,\langle[A,B]\rangle ,
\end{aligned}
\label{supp:eq:LP}
\end{equation}
which is the Lie--Poisson bracket \eqref{eq:LP} of the Letter. It is $U(r)$ gauge
invariant, as it must be, since $\langle A\rangle$ depends on $\widehat{\Psi}$ only
through $\rho=\widehat{\Psi}\widehat{\Psi}^\dagger$.

\paragraph{Charge algebra.}
With $A=J_a$, $B=J_b$ and the $su(2)$ commutators of \eqref{supp:eq:model},
\begin{equation}
\{\langle J_a\rangle,\langle J_b\rangle\}=\tfrac1\hbar\,\epsilon_{abc}\langle J_c\rangle,
\qquad \{\langle H\rangle,\langle J_a\rangle\}=0 .
\label{supp:eq:Jalg}
\end{equation}
We prescribe $\langle\J\rangle=q\,\n$ with $q\ge0$ and introduce the
shifted constraints $G_a\equiv\langle J_a\rangle-q\,n_a$. Since the shift
is a constant,
\begin{equation}
\{G_a,G_b\}=\tfrac1\hbar\,\epsilon_{abc}\,(G_c+q\,n_c),\qquad
\{\langle H\rangle,G_a\}=0 ,
\label{supp:eq:Galg}
\end{equation}
which is Eq.~\eqref{eq:Galg} of the Letter. On the constraint surface $G_a=0$ the
matrix of brackets is $C_{ab}=\{G_a,G_b\}|_{G=0}=(q/\hbar)\,\epsilon_{abc}n_c$.
In Dirac's terminology two cases arise.

\emph{$q=0$.} $C_{ab}=0$: all three constraints are first class and close
into $su(2)$. Their Hamiltonian vector fields generate the rotations
$\rho\to U\rho U^\dagger$, $U=e^{-i\boldsymbol\theta\cdot\J}$, which map the
constraint surface onto itself. At rotationally invariant states such as
$\rho_G$ all three vector fields vanish, so the zero-charge level is
singular there: the Poisson-algebra classification remains valid, but
regular constrained reduction, and a unique multiplier, cannot be inferred
at that stratum.

\emph{$q\neq0$.} $C_{ab}$ has rank two with null vector $n_a$. The
combination $\n\!\cdot\!\mathbf G$ is first class; it generates the
rotations about $\n$, which leave $\langle\J\rangle=q\n$ invariant and are
therefore tangent to the constraint surface. With $(\mathbf e_1,\mathbf
e_2,\n)$ a right-handed frame, the transverse pair obeys
$\{\mathbf e_1\!\cdot\!\mathbf G,\mathbf e_2\!\cdot\!\mathbf G\}\approx q/\hbar$
and is second class: a rotation about a transverse axis tilts the charge
and leaves the constraint surface.

\paragraph{The multiplier is parallel to the charge.}
Let $\mathcal O$ be the isospectral orbit of $\rho(0)$ and let
$\mathcal F(\bq)=\{\rho\in\mathcal O:\Tr(\rho J_a)=q_a\}$ be the feasible
set at charge vector $\bq$. For a rotation $U=e^{-i\boldsymbol\theta\cdot\J}$
with $SO(3)$ matrix $R$: (i) $U\rho U^\dagger\in\mathcal O$; (ii)
$\Tr(U\rho U^\dagger H)=\Tr(\rho H)$ because $[H,U]=0$; (iii)
$\Tr(U\rho U^\dagger J_a)=R_{ab}\Tr(\rho J_b)$. Hence $U$ maps
$\mathcal F(\bq)$ bijectively onto $\mathcal F(R\bq)$ without changing the
energy, and
\begin{equation}
E_{\min}(R\bq)=E_{\min}(\bq)\quad\Longrightarrow\quad
E_{\min}=E_{\min}(q),\qquad q=|\bq| .
\end{equation}
Therefore $\partial E_{\min}/\partial q_a=E_{\min}'(q)\,n_a$ wherever
$E_{\min}$ is differentiable. For $q>0$, on a differentiable minimizing
branch with regular constraints, the unique multiplier of the Letter,
$\lambda^a_{\rm cr}=-\partial E_{\min}/\partial q_a$, is therefore
\begin{equation}
\lamcr=-E_{\min}'(q)\,\n .
\label{supp:eq:lampar}
\end{equation}
At static or singular members this remains a valid multiplier, but it need
not be the unique solution of $[H+\boldsymbol\lambda\!\cdot\!\J,\rho_*]=0$;
in particular the one-sided derivative of $E_{\min}(q)$ at $q=0$ is not a
vector gradient (Remark~\ref{rem:cusp}). By Eq.~\eqref{eq:crit} of the Letter the
relative equilibrium evolves as
$\rho_*(t)=e^{\frac{i}{\hbar}\lamcr\cdot\J\,t}\rho_*e^{-\frac{i}{\hbar}\lamcr\cdot\J\,t}$,
a rotation about $\n$ with angular velocity $E_{\min}'(q)/\hbar$. The motion
of the time crystal is thus along the orbit of the first-class constraint,
as it must be: a rotation about a transverse axis would carry the state
off the constraint surface.

\section{The Gibbs orbit}
\label{sec:gibbs}

For the spectrum of $\rho(0)$ we take that of the Gibbs state
$\rho_G=e^{-\beta H}/Z$. Its eigenvalues are
\begin{equation}
p_-=\frac{e^{\beta\varepsilon}}{Z}\ \text{(twice)},\qquad
p_0=\frac{1}{Z}\ \text{(four times)},\qquad
p_+=\frac{e^{-\beta\varepsilon}}{Z}\ \text{(twice)},
\end{equation}
with $Z=2e^{\beta\varepsilon}+4+2e^{-\beta\varepsilon}=8\cosh^2(\beta\varepsilon/2)$,
and the following combinations will be needed:
\begin{equation}
\begin{gathered}
E_G=\langle H\rangle_G=-2\varepsilon(p_--p_+)=-\varepsilon\tanh\tfrac{\beta\varepsilon}{2},\\[2pt]
p_--p_0=\frac{e^{\beta\varepsilon}-1}{Z},\qquad
p_0-p_+=\frac{1-e^{-\beta\varepsilon}}{Z},\qquad
p_--p_+=\tfrac12\tanh\tfrac{\beta\varepsilon}{2}.
\end{gathered}
\label{supp:eq:pops}
\end{equation}
Equation~\eqref{eq:cov1} of the Letter is a commutator evolution with a Hermitian,
state-dependent generator and therefore preserves every eigenvalue of
$\rho$; equivalently, its amplitude lift conserves $\widehat{\Psi}^\dagger\widehat{\Psi}$ in the
gauge $\mathcal A=0$, which is the content of Eq.~\eqref{eq:eom} of the Letter. The
covariant qLLG flow thus stays on the orbit $\mathcal O_G=\{U\rho_GU^\dagger\}$
of all density matrices with the Gibbs spectrum, and the von Neumann
entropy keeps its thermal value $S(\beta)=-\sum_kp_k\ln p_k$ throughout.
The states on this orbit have the Gibbs spectrum; only $\rho_G$ itself is
the Gibbs state. The problem
solved below is
\begin{equation}
E_{\min}(q)=\min\{\Tr(\rho H):\ \rho\in\mathcal O_G,\ \Tr(\rho\J)=q\,\n\}.
\label{supp:eq:problem}
\end{equation}

\section{Von Neumann's trace inequality}
\label{sec:trace}

\begin{theorem}[von Neumann \cite{vonNeumann1937}]
\label{thm:trace}
Let $\rho\ge0$ have eigenvalues $p_1\ge p_2\ge\dots\ge p_d$ and let $K$ be
Hermitian with eigenvalues $\kappa_1\le\kappa_2\le\dots\le\kappa_d$. Then
\begin{equation}
\sum_{k=1}^dp_k\kappa_k\ \le\ \Tr(\rho K)\ \le\ \sum_{k=1}^dp_k\kappa_{d+1-k}.
\label{supp:eq:vN}
\end{equation}
\end{theorem}
\begin{proof}
Write $\rho=\sum_kp_k\ket{k}\bra{k}$ and
$K=\sum_l\kappa_l\ket{\kappa_l}\bra{\kappa_l}$. Then
\begin{equation}
\Tr(\rho K)=\sum_{k,l}p_k\,D_{kl}\,\kappa_l,\qquad D_{kl}=|\langle k|\kappa_l\rangle|^2 .
\end{equation}
Both bases are orthonormal, so every row and every column of $D$ sums to
one: $D$ is doubly stochastic. By the Birkhoff--von Neumann theorem
\cite{Birkhoff1946} it is a convex combination of permutation matrices,
$D=\sum_\pi c_\pi\Pi_\pi$, $c_\pi\ge0$, $\sum_\pi c_\pi=1$, so that
$\Tr(\rho K)=\sum_\pi c_\pi\sum_kp_k\kappa_{\pi(k)}$ is a convex
combination of the sums $\sum_kp_k\kappa_{\pi(k)}$. By the rearrangement
inequality \cite{HLP} the smallest of these sums pairs the sequences in
opposite order and the largest pairs them in the same order, which gives
the two bounds.
\end{proof}

\begin{lemma}[equality]
\label{lem:equality}
The lower bound in \eqref{supp:eq:vN} is attained if and only if $\rho$ has an
eigenbasis $\{\ket{e_k}\}$ that is also an eigenbasis of $K$, with
$\rho\ket{e_k}=p_k\ket{e_k}$ and $K\ket{e_k}=\kappa_k\ket{e_k}$ for the
sequences ordered as in the theorem. In particular $[\rho,K]=0$. Inside a
degenerate eigenspace of $K$ the basis is arbitrary, and inside a set of
equal $p_k$ the assignment is arbitrary.
\end{lemma}
\begin{proof}
Equality requires every permutation with $c_\pi>0$ to be minimizing. If
all $p_k$ and all $\kappa_l$ are distinct the only minimizing permutation
is the identity, $D=\mathbb 1$, and the eigenvectors coincide. When some
$\kappa_l$ coincide, any permutation inside the degenerate block is also
minimizing, so $\ket{k}$ may be any orthonormal basis of that block; the
same holds for coincident $p_k$. Conversely, every state of the stated
form gives equality by direct evaluation.
\end{proof}

Two consequences are used in the following sections.

\begin{corollary}[passive value]
\label{cor:passive}
On the orbit $\mathcal O$ of a state with eigenvalues $p_1\ge\dots\ge p_d$,
\begin{equation}
\min_{\rho\in\mathcal O}\Tr(\rho K)=P[K]\equiv\sum_kp_k\kappa_k ,
\end{equation}
and the minimizers are the states of Lemma~\ref{lem:equality}. Following
\cite{PuszWoronowicz1978,Lenard1978} we call them passive for $K$: no
unitary lowers their $K$-expectation further.
\end{corollary}

\begin{corollary}[maximal charge on an orbit]
\label{cor:cap}
The upper bound in \eqref{supp:eq:vN} with $K=\n\!\cdot\!\J$ gives, for the
Gibbs spectrum and the multiset of $m$ values in \eqref{supp:eq:levels},
\begin{equation}
\max_{\rho\in\mathcal O_G}\langle\n\!\cdot\!\J\rangle
=p_-\!\left(\tfrac32+\tfrac12\right)+p_0\!\left(\tfrac12+\tfrac12-\tfrac12-\tfrac12\right)
+p_+\!\left(-\tfrac12-\tfrac32\right)=2(p_--p_+)=\tanh\tfrac{\beta\varepsilon}{2}.
\label{supp:eq:qmax}
\end{equation}
No state on the orbit carries a larger charge, so the feasible set
\eqref{supp:eq:problem} is empty for $q>\tanh(\beta\varepsilon/2)$. This limit is
set by the spectrum alone: it is $\frac32$ for a pure state, $1$ for the
ground-doublet spectrum $(\frac12,\frac12,0,\dots)$, and $0$ for the
maximally mixed state.
\end{corollary}

\section{The Lagrange-multiplier bound and the level crossing}
\label{sec:bound}

\paragraph{Weak duality.}
The constrained functional \eqref{eq:F} of the Letter, restricted to the orbit
where the spectral constraint is already satisfied, is
$F=\langle H\rangle+\lambda^a\big(\langle J_a\rangle-q_a\big)$. By
Sec.~\ref{sec:constraints} we may take $\boldsymbol\lambda=\lambda\,\n$.
On the constraint surface $F=\langle H\rangle$, and therefore, for every
feasible $\rho$ and every real $\lambda$,
\begin{equation}
\langle H\rangle=\Tr\!\big[\rho\,(H+\lambda\,\n\!\cdot\!\J)\big]-\lambda q
\ \ge\ P(\lambda)-\lambda q,\qquad
P(\lambda)\equiv P[H+\lambda\,\n\!\cdot\!\J],
\label{supp:eq:dual}
\end{equation}
by Corollary~\ref{cor:passive}. Hence
\begin{equation}
E_{\min}(q)\ \ge\ \sup_{\lambda\in\mathbb R}\big[P(\lambda)-\lambda q\big].
\label{supp:eq:weak}
\end{equation}
The operator $H+\lambda\,\n\!\cdot\!\J$ is the co-moving Hamiltonian
$H+\boldsymbol\lambda\!\cdot\!\J$ of Eq.~\eqref{eq:minimum} of the Letter, so the bound
is nothing but the constrained functional evaluated with the trace
inequality.

\paragraph{Properties of $P(\lambda)$.}
$P(\lambda)$ is the minimum over the fixed set $\mathcal O_G$ of the
affine functions $\lambda\mapsto\Tr[\rho(H+\lambda\,\n\!\cdot\!\J)]$. It is
therefore concave and, since the passive state changes only at finitely
many level crossings, piecewise linear. Where $P$ is differentiable, every passive minimizer has the same charge,
and by the envelope theorem
\begin{equation}
P'(\lambda)=\langle\n\!\cdot\!\J\rangle_{\rm passive}(\lambda)\equiv q_{\rm pass}(\lambda),
\end{equation}
the common charge of the passive states (differentiability fixes this
charge, not the state itself: degenerate eigenspaces of $K$ can carry
distinct passive states). At a kink the one-sided derivatives are the
charges of the passive states on either side. The right-hand side of
\eqref{supp:eq:dual} therefore has slope $q_{\rm pass}(\lambda)-q$, and its
maximum is reached where the charge of the passive state crosses the
prescribed $q$. Since $P(\lambda)$ is even under $\lambda\to-\lambda$
(replace $\n$ by $-\n$), it suffices to discuss $\lambda\le0$ for $q>0$.

\paragraph{The levels of the co-moving Hamiltonian.}
In the joint eigenbasis the levels of $H+\lambda\,\n\!\cdot\!\J$ are
$E+\lambda m$. For the Gibbs spectrum the populations are equal within each
multiplet, so as long as the three multiplets do not interleave the
passive state is $\rho_G$ itself: the splittings $\lambda m$ cancel in
$\sum_{m}p\,(E+\lambda m)$ because $\sum_mm=0$ within a multiplet. The
multiplets first touch when the highest level of the quadruplet reaches
the lowest level of the excited doublet and, simultaneously, when the
highest level of the ground doublet reaches the lowest level of the
quadruplet. For $\lambda=-\mu<0$:
\begin{equation}
\begin{array}{lll}
\text{pair }A:\quad & \kE{-}{-\tfrac12}:\ -\varepsilon+\tfrac{\mu}{2},
 & \kE{\tfrac32}{\tfrac32}:\ -\tfrac{3\mu}{2},\\[3pt]
\text{pair }B:\quad & \kE{\tfrac32}{-\tfrac32}:\ +\tfrac{3\mu}{2},
 & \kE{+}{+\tfrac12}:\ \varepsilon-\tfrac{\mu}{2},
\end{array}
\qquad\text{both cross at }\mu=\tfrac{\varepsilon}{2},
\label{supp:eq:pairs}
\end{equation}
at the common values $-\frac34\varepsilon$ and $+\frac34\varepsilon$. Both
pairs have $\Delta m=2$ and $\Delta E=\varepsilon$. All further crossings
between levels of \emph{unequal} population would occur at $\mu=\varepsilon$
or later, but by then the populations of the crossing levels have become
equal through the exchange of $A$ and $B$ (see below), so they change
nothing. Hence
\begin{equation}
P(\lambda)=\begin{cases}
E_G, & |\lambda|\le\varepsilon/2,\\[2pt]
E_G+\big(\tfrac{\varepsilon}{2}-|\lambda|\big)\,q_{\max}, & |\lambda|\ge\varepsilon/2,
\end{cases}
\qquad q_{\max}=\tanh\tfrac{\beta\varepsilon}{2},
\label{supp:eq:P}
\end{equation}
where for $|\lambda|>\varepsilon/2$ the passive state is the fully
exchanged state $\rho_{\rm cap}$ obtained from $\rho_G$ by swapping the
populations inside $A$ and inside $B$; it has charge $q_{\max}$ and energy
$E_G+\varepsilon q_{\max}/2$.

\paragraph{Selection of the multiplier.}
For $-\varepsilon/2\le\lambda\le0$ the right-hand side of \eqref{supp:eq:dual} is
$E_G-\lambda q$, which grows as $\lambda$ decreases; for
$\lambda<-\varepsilon/2$ it is
$E_G+\frac{\varepsilon}{2}q_{\max}+\lambda\,(q_{\max}-q)$, which decreases
as $\lambda$ decreases further because $q\le q_{\max}$. The supremum in
\eqref{supp:eq:weak} is therefore attained at the crossing,
\begin{equation}
\lambda_{\rm cr}=-\tfrac{\varepsilon}{2},\qquad
E_{\min}(q)\ \ge\ E_G+\tfrac{\varepsilon}{2}\,q\qquad(0\le q\le q_{\max}).
\label{supp:eq:bound}
\end{equation}
Physically, $\varepsilon/2$ is the energy cost per unit of charge of the
cheapest rearrangement of weight on the orbit. Exchanging the populations
of two levels with $\Delta m$ and $\Delta E$ costs $\Delta E/\Delta m$ per
unit of charge: the pairs $A$ and $B$ cost $\varepsilon/2$; every other
pair of levels with unequal populations has $\Delta m=1$ at
$\Delta E=\varepsilon$ or $\Delta m\le1$ at $\Delta E=2\varepsilon$, i.e.\
at least $\varepsilon$ per unit of charge; and exchanges inside a
multiplet carry no charge because the populations there are equal.

For $q=0$ the right-hand side of \eqref{supp:eq:dual} is maximal on the whole
plateau $|\lambda|\le\varepsilon/2$: the multiplier is undetermined. For
$q=q_{\max}$ it is maximal for all $\lambda\le-\varepsilon/2$. Both cases
correspond to static minimizers (Sec.~\ref{sec:limits}).

\section{The minimizers and the result}
\label{sec:minimizers}

\paragraph{Passive states of the co-moving Hamiltonian.}
Equality in \eqref{supp:eq:bound} requires, by Lemma~\ref{lem:equality}, a
passive state of $K_{\rm cr}\equiv H-\frac{\varepsilon}{2}\,\n\!\cdot\!\J$.
Its levels and the populations they receive are
\begin{equation}
\begin{array}{c|cccccccc}
\text{level of }K_{\rm cr} & -\tfrac54\varepsilon & -\tfrac34\varepsilon & -\tfrac34\varepsilon
 & -\tfrac14\varepsilon & +\tfrac14\varepsilon & +\tfrac34\varepsilon & +\tfrac34\varepsilon & +\tfrac54\varepsilon\botstrut\\ \hline
\topstrut\text{state} & \kE{-}{+\tfrac12} & \multicolumn{2}{c}{\underbrace{\textstyle\kE{-}{-\tfrac12},\ \kE{\tfrac32}{\tfrac32}}_{A}}
 & \kE{\tfrac32}{\tfrac12} & \kE{\tfrac32}{-\tfrac12}
 & \multicolumn{2}{c}{\underbrace{\textstyle\kE{\tfrac32}{-\tfrac32},\ \kE{+}{+\tfrac12}}_{B}} & \kE{+}{-\tfrac12}\\[2pt]
\text{population} & p_- & p_- & p_0 & p_0 & p_0 & p_0 & p_+ & p_+
\end{array}
\label{supp:eq:passiveK}
\end{equation}
The passive value is
$P(-\varepsilon/2)=p_-(-\frac54-\frac34)\varepsilon+p_0(-\frac34-\frac14+\frac14+\frac34)\varepsilon
+p_+(\frac34+\frac54)\varepsilon=-2\varepsilon(p_--p_+)=E_G$, as used in
\eqref{supp:eq:P}. The two degenerate levels leave the basis inside $A$ and
inside $B$ free. The general passive state is therefore
\begin{equation}
\begin{aligned}
\rho_*&=p_-\,\big(\PE{-}{+\frac12}+\ket{a}\bra{a}\big)\\
&\quad+p_0\,\big(\ket{a'}\bra{a'}+\PE{\frac32}{\frac12}+\PE{\frac32}{-\frac12}+\ket{b}\bra{b}\big)\\
&\quad+p_+\,\big(\ket{b'}\bra{b'}+\PE{+}{-\frac12}\big),
\qquad \PE{\mu}{m}\equiv\kE{\mu}{m}\bra{\mu,m}_{\n},
\end{aligned}
\label{supp:eq:rhostar}
\end{equation}
with
\begin{equation}
\begin{aligned}
\ket{a}&=\cos\theta\,\kE{-}{-\tfrac12}+e^{i\phi}\sin\theta\,\kE{\tfrac32}{\tfrac32},
&\ket{a'}&=-e^{-i\phi}\sin\theta\,\kE{-}{-\tfrac12}+\cos\theta\,\kE{\tfrac32}{\tfrac32},\\
\ket{b}&=\cos\vartheta\,\kE{\tfrac32}{-\tfrac32}+e^{i\varphi}\sin\vartheta\,\kE{+}{+\tfrac12},
&\ket{b'}&=-e^{-i\varphi}\sin\vartheta\,\kE{\tfrac32}{-\tfrac32}+\cos\vartheta\,\kE{+}{+\tfrac12},
\end{aligned}
\label{supp:eq:family}
\end{equation}
where $\theta,\vartheta\in[0,\pi/2]$ and $\phi,\varphi$ are arbitrary
phases; $\ket{a}$ and $\ket{b}$ are the states \eqref{eq:family2} of the Letter. Since
$\ket{a}\bra{a}+\ket{a'}\bra{a'}$ is the projector on pair $A$, and
likewise for $B$, \eqref{supp:eq:rhostar} is the form \eqref{eq:family1} of the Letter,
$\rho_*=\rho_G+(p_--p_0)\big(\ket a\bra a-\PE{-}{-\frac12}\big)
+(p_0-p_+)\big(\ket b\bra b-\PE{\frac32}{-\frac32}\big)$.

\paragraph{Charge and energy of the family.}
Inside pair $A$ the populations $p_-$ and $p_0$ sit on $\ket{a}$ and
$\ket{a'}$. The contribution of the pair to $\langle\n\!\cdot\!\J\rangle$ is
$p_-(-\frac12\cos^2\theta+\frac32\sin^2\theta)+p_0(-\frac12\sin^2\theta+\frac32\cos^2\theta)$,
which exceeds its value at $\theta=0$ by $2(p_--p_0)\sin^2\theta$, and to
$\langle H\rangle$ by $\varepsilon(p_--p_0)\sin^2\theta$; pair $B$
contributes analogously with $p_0-p_+$ and $\vartheta$. All other terms of
\eqref{supp:eq:rhostar} are those of $\rho_G$. The coherences in \eqref{supp:eq:family}
connect states with $\Delta m=2$, while $J_\pm$ change $m$ by one, so the
transverse charge vanishes. Hence
\begin{equation}
\langle\J\rangle_*=q\,\n,\qquad
q=2(p_--p_0)\sin^2\!\theta+2(p_0-p_+)\sin^2\!\vartheta,\qquad
\langle H\rangle_*=E_G+\tfrac{\varepsilon}{2}\,q ,
\label{supp:eq:qE}
\end{equation}
the second of which is the constraint \eqref{eq:angle_constraint} of the Letter. As
$(\theta,\vartheta)$ range over $[0,\pi/2]^2$, $q$ covers
$[0,\,2(p_--p_+)]=[0,q_{\max}]$. Every charge in this window is therefore
attained with the energy $E_G+\varepsilon q/2$, and the bound
\eqref{supp:eq:bound} is saturated.

\begin{theorem}
\label{thm:main}
On the Gibbs orbit $\mathcal O_G$ at inverse temperature $\beta$,
\begin{equation}
E_{\min}(q)=\varepsilon\Big[\frac q2-\tanh\frac{\beta\varepsilon}{2}\Big]
\qquad\text{for }0\le q\le\tanh\frac{\beta\varepsilon}{2},
\label{supp:eq:Emin}
\end{equation}
and there is no state with $q>\tanh(\beta\varepsilon/2)$. The minimizers at
charge $q$ are exactly the states \eqref{supp:eq:rhostar}--\eqref{supp:eq:family}
with $(\theta,\vartheta)$ on the curve \eqref{supp:eq:qE} and arbitrary
$\phi,\varphi$.
\end{theorem}
\begin{proof}
The lower bound is \eqref{supp:eq:bound}, attainment is \eqref{supp:eq:qE}, and the
window is Corollary~\ref{cor:cap}. A minimizer must give equality in
\eqref{supp:eq:dual} at $\lambda=-\varepsilon/2$ and is therefore passive for
$K_{\rm cr}$ by Lemma~\ref{lem:equality}, i.e.\ of the form
\eqref{supp:eq:rhostar}.
\end{proof}

\begin{remark}[duality and the cusp]
\label{rem:cusp}
On the Gibbs orbit, Eqs.~\eqref{supp:eq:P} and \eqref{supp:eq:Emin} are
concave/convex conjugates of each other: since $q=|\langle\J\rangle|\ge0$,
\begin{equation}
P(\lambda)=\min_{0\le q\le q_{\max}}\big[E_{\min}(q)-|\lambda|\,q\big],
\qquad
E_{\min}(q)=\sup_{\lambda\in\mathbb R}\big[P(\lambda)-\lambda q\big]
\quad(0\le q\le q_{\max}),
\label{supp:eq:conj}
\end{equation}
which is the content of the relation
$\lambda_{\rm cr}=-\partial E_{\min}/\partial q$ of the Letter. This strong
duality is a property of the Gibbs orbit; it is not claimed for arbitrary
prescribed spectra (Sec.~\ref{sec:limits}). As a function of the charge
vector, $E_{\min}(\bq)=E_G+\frac{\varepsilon}{2}|\bq|$ is a cone with its
apex at $\bq=0$. Its subdifferential there is the ball
$|\boldsymbol\lambda|\le\varepsilon/2$, which is the plateau of $P$. Two
different notions must be kept apart at $q=0$: stationarity of $\rho_G$,
$[H,\rho_G]=[J_a,\rho_G]=0$, allows \emph{every} $\boldsymbol\lambda$ in
Eq.~\eqref{eq:minimum} of the Letter, whereas the optimal multipliers of the bound
\eqref{supp:eq:dual} form the ball. At the exact Gibbs state the charge-conservation
system that determines the multipliers [Eq.~\eqref{eq:charge} of the Letter;
Sec.~\ref{sec:numerics}] is satisfied identically for every
$\boldsymbol\lambda$, and its minimum-norm solution is $\boldsymbol\lambda=0$. The value observed along a trajectory
that approaches $\rho_G$ is a property of the path: for the random initial
states tabulated in Sec.~\ref{sec:numerics} the multiplier approaches the
boundary of the ball, $|\boldsymbol\lambda|\to\varepsilon/2$, in a
direction set by the initial state, whereas for the initial state
$\big[U_\chi\,{\rm diag}(p_-,p_+)U_\chi^\dagger\otimes I_2\big]\oplus p_0I_4$
(a rotation inside chirality space at fixed total spin), which has $q=0$
and $[J_a,\rho(t)]=0$ throughout, the flow relaxes to $\rho_G$ with
$\boldsymbol\lambda\equiv0$. No limiting value of the multiplier at $q=0$
is therefore an invariant of the endpoint.
\end{remark}

\paragraph{Uniqueness of the multiplier.}
Equation~\eqref{eq:minimum} of the Letter, $[H+\lamcr\!\cdot\!\J,\rho_*]=0$, holds for
$\lamcr=-\frac{\varepsilon}{2}\n$ by construction, since $\rho_*$ is
block-diagonal in the eigenspaces of $K_{\rm cr}$. The multiplier is unique
when the three commutators $[J_a,\rho_*]$ are linearly independent. A
vanishing combination $[\mathbf c\!\cdot\!\J,\rho_*]=0$ means invariance
of $\rho_*$ under rotations about $\mathbf c$. For $q\neq0$ and $\mathbf c$
not parallel to $\n$ such a rotation tilts the charge, so $\rho_*$ cannot be
invariant. For $\mathbf c=\n$ the rotation multiplies the coherence inside
$A$ and $B$ by $e^{-2i\alpha}$, so $\rho_*$ is invariant only if both
coherences vanish, $\sin\theta\cos\theta=\sin\vartheta\cos\vartheta=0$.
For finite positive $\beta$, where $p_->p_0>p_+$, the commutators are
therefore independent, and $\lamcr=-\frac{\varepsilon}{2}\n$ is unique, for
every member of the family except the four diagonal ones,
$(\theta,\vartheta)\in\{0,\pi/2\}^2$. Equivalently, a member of the family
is moving exactly when
\begin{equation}
(p_--p_0)\sin2\theta\neq0\qquad\text{or}\qquad(p_0-p_+)\sin2\vartheta\neq0,
\label{supp:eq:moving}
\end{equation}
a criterion that remains valid in the limits of Sec.~\ref{sec:limits}
(at $T=0$ the angle $\vartheta$ does not change the state). The diagonal
members occur at the charges
\begin{equation}
q\in\{0,\ q_A,\ q_B,\ q_A+q_B\},\qquad q_A=2(p_--p_0),\qquad q_B=2(p_0-p_+),
\qquad q_A+q_B=q_{\max},
\label{supp:eq:static}
\end{equation}
and are static. This is the qualification ``generic $(\theta,\vartheta)$''
of the Letter. At the static members $\lamcr$ remains the unique interior
optimizer of the dual bound \eqref{supp:eq:weak}, but stationarity alone,
Eq.~\eqref{eq:minimum} of the Letter, admits additional multipliers there. At the
interior values $q_A$ and $q_B$ the static members are degenerate in energy
with the moving members of the same charge; the sampled initial states at
exactly these charges converge to moving members (Table~\ref{tab:series}).

\section{Motion of the minimizing family and its signatures}
\label{sec:motion}

\paragraph{Rotation.}
With $\lamcr=-\frac{\varepsilon}{2}\n$, Eq.~\eqref{eq:crit} of the Letter gives
\begin{equation}
\rho_*(t)=e^{-\frac{i}{\hbar}\frac{\varepsilon}{2}t\,\n\cdot\J}\,\rho_*\,
e^{+\frac{i}{\hbar}\frac{\varepsilon}{2}t\,\n\cdot\J},
\end{equation}
a rotation about $\n$ with angular velocity $\varepsilon/2\hbar=\sqrt3\,\omega/2$.
Equivalently, under $e^{-iHt/\hbar}$ every matrix element
$\bra{k}\rho\ket{l}$ acquires the phase $e^{-i(E_k-E_l)t/\hbar}$; the only
off-diagonal elements of $\rho_*$ in the joint basis are the coherences
inside $A$ and $B$, which connect levels with $E_k-E_l=\mp\varepsilon$. In
the states \eqref{supp:eq:family} this is the advance of the azimuths,
$\phi(t)=\phi-\varepsilon t/\hbar$ and $\varphi(t)=\varphi-\varepsilon t/\hbar$,
at constant $\theta$, $\vartheta$ and constant difference $\phi-\varphi$.
A rotation by $\pi$ about $\n$ multiplies a $\Delta m=2$ coherence by
$e^{-2\pi i}=1$ and leaves the diagonal part invariant, so for a
nonstationary member [Eq.~\eqref{supp:eq:moving}] the fundamental period of the
state is
\begin{equation}
T=\frac{2\pi\hbar}{\varepsilon}=\frac{2\pi}{\sqrt3\,\omega},
\end{equation}
half the period of a full $2\pi$ rotation; the four diagonal members of
\eqref{supp:eq:static} are static.

\paragraph{Single-spin expectations are static.}
$\boldsymbol\sigma_i$ changes $m$ by $0$ or $\pm1$ and therefore does not
connect the two components of $\ket{a}$ or of $\ket{b}$, which differ by
$\Delta m=2$. Only the diagonal part of $\rho_*$ contributes, and by
Lemma~\ref{lem:single}
\begin{equation}
\langle\boldsymbol\sigma_i\rangle_*(t)=\frac23\,\n\sum_k m_k\,\bra{k}\rho_*\ket{k}=\frac{2q}{3}\,\n
\qquad(i=1,2,3),
\end{equation}
where $\ket k$ runs over the joint basis \eqref{supp:eq:levels}, whose diagonal
populations depend on the mixing angles (they are not the eigenvalues
$p_k$ of $\rho_*$), and $\sum_km_k\bra k\rho_*\ket k=\langle\n\!\cdot\!\J\rangle=q$.
The result is the same on all sites and constant in time. The magnetizations of the time
crystal are thus those of a static, uniformly polarized cluster; at $q=0$
they vanish identically.

\paragraph{Double-raising two-spin correlations oscillate.}
With raising operators $\sigma_i^\pm=\frac12(\sigma_i^{x}\pm i\sigma_i^{y})$
defined relative to $\n$, the operator $\sigma_i^+\sigma_j^+$ ($i\neq j$)
raises $m$ by two. On the full eight-dimensional space it connects further
pairs of states, but for the family \eqref{supp:eq:rhostar} only the coherences
inside $A$ and $B$ contribute to its expectation value. From \eqref{supp:eq:rhostar},
$\bra{-,-\tfrac12}\rho_*\kE{\tfrac32}{\tfrac32}=(p_--p_0)\sin\theta\cos\theta\,e^{-i\phi}$
and
$\bra{\tfrac32,-\tfrac32}\rho_*\kE{+}{+\tfrac12}=(p_0-p_+)\sin\vartheta\cos\vartheta\,e^{-i\varphi}$,
so that
\begin{equation}
\langle\sigma_i^+\sigma_j^+\rangle_*(t)=e^{\,i\varepsilon t/\hbar}\Big[
(p_--p_0)\sin\theta\cos\theta\,e^{-i\phi}M_A^{ij}
+(p_0-p_+)\sin\vartheta\cos\vartheta\,e^{-i\varphi}M_B^{ij}\Big],
\label{supp:eq:twospin}
\end{equation}
with $M_A^{ij}=\bra{\tfrac32,\tfrac32}\sigma_i^+\sigma_j^+\kE{-}{-\tfrac12}$
and $M_B^{ij}=\bra{+,\tfrac12}\sigma_i^+\sigma_j^+\kE{\tfrac32}{-\tfrac32}$.
From the explicit states of Sec.~\ref{sec:model},
$\sigma_1^+\sigma_2^+\kE{-}{-\frac12}=\frac{1}{\sqrt3}\ket{\uparrow\uparrow\uparrow}$,
and in general $|M_A^{ij}|=|M_B^{ij}|=1/\sqrt3$ for all three pairs of
sites; with the phase conventions of Sec.~\ref{sec:model},
$M_B^{ij}=-M_A^{ij}$. Both terms oscillate at the same angular frequency
$\varepsilon/\hbar=\sqrt3\,\omega$, so whenever its amplitude is nonzero
the signal is a single complex exponential with period $T$, whose real and
imaginary parts are sinusoids. Its amplitude depends on how the charge is
shared between the pairs, $q=q_A\sin^2\theta+q_B\sin^2\vartheta$, and on
the relative phase $\phi-\varphi$:
\begin{equation}
\big|\langle\sigma_i^+\sigma_j^+\rangle_*\big|
=\frac{1}{2\sqrt3}\,\Big|\sqrt{q_a(q_A-q_a)}\;+\;e^{i\delta}\sqrt{q_b(q_B-q_b)}\Big|,
\qquad q_a=q_A\sin^2\!\theta,\quad q_b=q_B\sin^2\!\vartheta,
\label{supp:eq:amp}
\end{equation}
where $q_a+q_b=q$, $q_A=2(p_--p_0)$, $q_B=2(p_0-p_+)$, and
$\delta=\phi-\varphi+\pi$ is the relative phase, including that of
$M_B^{ij}/M_A^{ij}=-1$. The oscillation frequency is
$\Delta E/\hbar=\varepsilon/\hbar$, whereas the symmetry-rotation rate is
$(\Delta E/\Delta m)/\hbar=\varepsilon/2\hbar$. Temperature, charge and the
internal allocation determine the amplitude. Its envelope vanishes as
$O(\sqrt q)$ near $q=0$ and as $O(\sqrt{q_{\max}-q})$ near the cap, and at
fixed $q$ the amplitude ranges over a band, between the minimum of
$|\sqrt{q_a(q_A-q_a)}-\sqrt{q_b(q_B-q_b)}|/2\sqrt3$ and the maximum of
$[\sqrt{q_a(q_A-q_a)}+\sqrt{q_b(q_B-q_b)}]/2\sqrt3$ over the admissible
splits $q_a+q_b=q$. Destructive interference can suppress an individual
signal further: the two terms of \eqref{supp:eq:twospin} cancel for equal
coherence magnitudes and $\phi=\varphi$. At $\beta\varepsilon=1$ and
$q=0.05$, for instance, $\theta=0.1751353550$ and $\vartheta=0.6010307390$
give $q_a=0.01025668677$, $q_b=0.03974331323$ and two coherences of equal
magnitude $0.02898217210$, and every $\langle\sigma_i^+\sigma_j^+\rangle$
vanishes. All one- and two-spin expectations of this member are then time
independent although $[H,\rho_*]\neq0$; its motion shows in the three-spin
probe $\langle\sigma_1^+\sigma_2^+\sigma_3^{z}\rangle$, which oscillates
with amplitude $0.03346572973$. The motion is therefore generically, not
unconditionally, visible in two-spin correlations, consistent with the
zero lower edge of several amplitude bands in Tables~\ref{tab:attractors}
and \ref{tab:series}. A density-matrix measure of the motion that does not
depend on the choice of probe is
\begin{equation}
\|\dot\rho_*\|_{HS}^2=\frac{\varepsilon^2}{2\hbar^2}\Big[q_a(q_A-q_a)+q_b(q_B-q_b)\Big],
\label{supp:eq:rdot}
\end{equation}
which vanishes exactly when both block coherences vanish, i.e.\ exactly
for the static members \eqref{supp:eq:static}. The split $(q_a,q_b)$ and the
phases are not fixed by the minimization and are selected by the initial
state of the flow; this is why runs with the same $q$ in
Sec.~\ref{sec:numerics} end with different amplitudes but the same
frequency. The last two columns of Table~\ref{tab:attractors} compare the
amplitude reached in each run with this band, and Table~\ref{tab:series}
compares $\|\dot\rho\|$ with \eqref{supp:eq:rdot}.

\section{Limits and special cases}
\label{sec:limits}

\paragraph{$q=0$: the Gibbs state.}
At $q=0$ the family reduces to $\theta=\vartheta=0$, i.e.\ to $\rho_G$,
the unique passive state of $H$ on its own orbit. It is $SU(2)$ invariant,
$[J_a,\rho_G]=0$, so Eq.~\eqref{eq:minimum} of the Letter holds for every
$\boldsymbol\lambda$ and the multipliers are undetermined, in agreement
with the plateau of $P$. The sampled generic initial states at $q=0$
(Sec.~\ref{sec:numerics}) converge to $\rho_G$. This is not a statement
about every initial state: the energy decreases monotonically along the
flow, which is a Lyapunov statement, but the orbit contains other fixed
points. The state $\rho_{\rm inv}=p_+P_-+p_0P_Q+p_-P_+$, with $P_\mu$ the
projectors on the three multiplets, has the Gibbs spectrum,
$\langle\J\rangle=0$ and the energy $-E_G>E_G$; it commutes with $H$ and
with all $J_a$, so the covariant flow leaves it fixed. That the
zero-charge minimizer is $\rho_G$ and is stationary is consistent with the
no-go theorems for ground and Gibbs states, which the construction thus
recovers as its zero-charge member without deriving them. If the flow
reaches the constrained minimum, its energy decrease is
$E(\rho_0)-E_G-\varepsilon q/2$. At $q=0$ this equals, as a numerical
energy difference, the ergotropy \cite{Allahverdyan2004} of the initial
state (whether the decrease is extractable as work depends on a physical
work-storage or control implementation, which the phenomenological
dissipative equation does not supply); the energy $\varepsilon q/2$ that
cannot be removed without changing the charge is the energy of the
rotating state.

\paragraph{$q=q_{\max}$: the fully exchanged state.}
At $q=q_{\max}$ the constraint \eqref{supp:eq:qE} forces
$\theta=\vartheta=\pi/2$, so the energy-minimizing state is uniquely
$\rho_{\rm cap}$, in which both pairs are completely exchanged. It is not
the only state of this charge on the orbit: rotating $\kE{-}{+\frac12}$
against $\kE{\frac32}{\frac12}$ inside their common $m=\frac12$ subspace
by an angle $\gamma$ preserves the spectrum and the charge but raises the
energy by $\varepsilon(p_--p_0)\sin^2\gamma$, because the two
charge-maximizing eigenspaces of $\n\!\cdot\!\J$ carry unequal
populations. The state $\rho_{\rm cap}$ is diagonal in the joint basis,
hence static, with $E=E_G+\frac{\varepsilon}{2}q_{\max}=-\frac{\varepsilon}{2}\tanh(\beta\varepsilon/2)$.
Since $E_G=-\varepsilon q_{\max}$, the fully exchanged states of all
temperatures lie on the line $E=-\varepsilon q/2$ of Fig.~\ref{fig:gibbs}(a) of the
Letter.

\paragraph{Zero temperature.}
For $\beta\to\infty$, $p_0,p_+\to0$ and $p_-\to\frac12$. The family
reduces to the rank-two states
\begin{equation}
\rho_*=\tfrac12\,\PE{-}{+\frac12}+\tfrac12\,\ket{a}\bra{a},\qquad
q=\sin^2\theta\in[0,1],\qquad E_{\min}(q)=-\varepsilon+\tfrac{\varepsilon}{2}q,
\end{equation}
on the orbit of the ground-doublet spectrum $(\frac12,\frac12,0,\dots,0)$.
At $q=0$ this is the static state $P_-/2$; for $0<q<1$ it is a rank-two
time crystal; at $q=1$ it is again static. This is the $T=0$ line of
Fig.~\ref{fig:gibbs}(a) of the Letter.

\paragraph{Infinite temperature.}
For $\beta\to0$ all populations tend to $\frac18$, $q_{\max}\to0$, and the
window closes: the maximally mixed state is the only state on its orbit.

\paragraph{Other spectra.}
The static minimum at $q=0$ is a property of the Gibbs spectrum, or more
generally of any spectrum whose passive state for $H$ is a function of
$H$, hence $SU(2)$ invariant; this is the case when the eigenvalues
assigned to each eigenspace of $H$ are equal. For other spectra a
zero-charge minimum may be moving or stationary, depending on the
spectrum. A stationary example that is not a function of $H$ is the state
diagonal in the joint basis with weight $0.22$ on each ground-doublet
state, $0.16$ on the quadruplet states with $m=\pm\frac32$, $0.10$ on those
with $m=\pm\frac12$ and $0.02$ on each excited-doublet state: it has zero
charge, is passive for $H$ and commutes with $H$ and $\n\!\cdot\!\J$, hence
is static, but its quadruplet populations differ. A moving example is a
pure state, whose passive state $\kE{-}{+\frac12}$ has $q=\frac12$, so
that no eigenstate meets the constraint $q=0$. Let $x$, $y$, $z$ be the
weights of a pure state in the negative-chirality doublet, the
positive-chirality doublet and the quadruplet. Since $\J$ has no matrix
elements between these sectors and the normalized doublet components have
spin length $\frac12$, $\langle\J\rangle=0$ requires
$\frac x2\le\frac y2+\frac{3z}{2}$ with $x+y+z=1$. With $w=x-y$ this gives
$w\le3z$ and $w\le1-z$, hence $w\le\frac34$ and
$E=-\varepsilon w\ge-\frac34\varepsilon$, with equality for the
superposition of $\kE{-}{+\frac12}$ with weight $\frac34$ and
$\kE{\frac32}{-\frac32}$ with weight $\frac14$. This is the constrained
minimum at $q=0$, $E_{\min}=-\frac34\varepsilon$. The level crossing of
the two states involved, a $\Delta m=2$, $\Delta E=\varepsilon$ pair,
identifies this relative equilibrium, with $|\lamcr|=\varepsilon/2$ and
the frequency $\sqrt3\,\omega$; but the state is not passive for
$H+\lamcr\!\cdot\!\J$, since at the crossing another level lies lower, and
$\sup_\lambda\lambda_{\min}(H+\lambda\,\n\!\cdot\!\J)=-\varepsilon<-\frac34\varepsilon$.
For a pure state the bound \eqref{supp:eq:weak} therefore has a duality gap,
and minimality rests on the direct inequality above, not on
Eq.~\eqref{supp:eq:conj}. The frequency of the Gibbs-orbit family is not
claimed for other spectra beyond this example.

\section{Isotropic exchange and the doublet-mixing branch}
\label{sec:exchange}

\paragraph{Exchange gap.}
The dominant interaction in $\mathrm{Cu}_3$ triangles is the
antiferromagnetic Heisenberg exchange
$H_{\rm ex}=J\sum_{i<j}\mathbf s_i\!\cdot\!\mathbf s_j
=\frac{J}{2}(\J^2-\frac94)$, $J>0$, with $\mathbf s_i=\boldsymbol\sigma_i/2$.
It commutes with $H$ and with $\J$, equals $-\frac34J$ on both doublets and
$+\frac34J$ on the quadruplet, and therefore raises the quadruplet by the
exchange gap $\Delta=\frac32J$ relative to the doublets, while the
chirality splitting $2\varepsilon$ between the doublets is untouched. We
measure energies from the centre of the two doublets, i.e.\ we use
\begin{equation}
H_\Delta=H+H_{\rm ex}+\tfrac34J\,I=H+\Delta\,P_Q,\qquad \Delta=\tfrac32J,
\label{supp:eq:HDelta}
\end{equation}
with $P_Q$ the projector on the quadruplet, so that the levels are
$-\varepsilon$, $+\varepsilon$ and $\Delta$; this shifted zero is used in
all Gibbs weights and in $E_G$ below. The raw energy differs by the
constant $-\frac34J$, which changes neither the normalized Gibbs state nor
the dynamics. (The value of $\Delta$ in terms of the exchange constant
depends on the convention for $H_{\rm ex}$; the gap $\Delta$ itself does
not.) Since $H_{\rm ex}$ is a function of $\J^2$ alone, the conservation
structure, the joint eigenbasis, Lemma~\ref{lem:single} and the whole
construction of Secs.~\ref{sec:trace}--\ref{sec:minimizers} carry over
verbatim; only the level positions $E$ in \eqref{supp:eq:levels}, and through
them the thermal populations and the crossing sequence, change. The Gibbs
weights become
\begin{equation}
p_\pm=\frac{e^{\mp\beta\varepsilon}}{Z},\qquad p_0=\frac{e^{-\beta\Delta}}{Z},\qquad
Z=4\cosh\beta\varepsilon+4e^{-\beta\Delta},\qquad
E_G=\frac{\Delta e^{-\beta\Delta}-\varepsilon\sinh\beta\varepsilon}{\cosh\beta\varepsilon+e^{-\beta\Delta}},
\label{supp:eq:popsex}
\end{equation}
with $p_->p_+>p_0$ as soon as $\Delta>\varepsilon$. Relevant candidate
level crossings and their signed ratios $\Delta E/\Delta m$ are
\begin{equation}
\begin{array}{lccc}
 & \Delta m & \Delta E & \Delta E/\Delta m\botstrut\\ \hline
\topstrut A:\ \kE{-}{-\tfrac12}\leftrightarrow\kE{\tfrac32}{\tfrac32} & 2 & \varepsilon+\Delta & (\varepsilon+\Delta)/2\\[2pt]
B:\ \kE{\tfrac32}{-\tfrac32}\leftrightarrow\kE{+}{+\tfrac12} & 2 & \varepsilon-\Delta & (\varepsilon-\Delta)/2\\[2pt]
B':\ \kE{+}{-\tfrac12}\leftrightarrow\kE{\tfrac32}{\tfrac32} & 2 & \Delta-\varepsilon & (\Delta-\varepsilon)/2\\[2pt]
C:\ \kE{-}{-\tfrac12}\leftrightarrow\kE{+}{+\tfrac12} & 1 & 2\varepsilon & 2\varepsilon\\[2pt]
\text{other}\ \Delta m=1\ \text{pairs} & 1 & \varepsilon\pm\Delta,\ \Delta\pm\varepsilon & \varepsilon\pm\Delta,\ \Delta\pm\varepsilon
\end{array}
\label{supp:eq:costs}
\end{equation}
A positive-charge transfer moves the larger population of a pair toward
the larger $m$, at the cost $\Delta E/\Delta m$ per unit of charge; which
of the candidate pairs are actual transfers therefore depends on the
ordering of the populations, i.e.\ on $\Delta$. In particular pair $B$,
as written, is a positive-charge transfer only for $\Delta<\varepsilon$:
for $\Delta>\varepsilon$ its population ordering is reversed (the ratio
becomes negative) and $B'$ takes its place.

\paragraph{Strong exchange, $\Delta>5\varepsilon$: the doublet-mixing branch.}
This exchange-dominated parameter range is the one treated in the Letter.
It is motivated by the separation of energy scales in the $\mathrm{Cu}_3$
example of Ref.~\cite{Iida2011}: inelastic neutron scattering on that
polyoxometalate triangle places the four lowest levels at approximately
$0$, $0.106$, $0.592$ and $0.622$~meV, so that identifying the $0.106$~meV
splitting of the two lowest doublets with $2\varepsilon$, for the purpose
of the comparison only, puts the mean position of the upper levels at
$\Delta\approx0.55$~meV above the centre of the doublets, i.e.\
$\Delta/\varepsilon\approx10$. Ref.~\cite{Iida2011} attributes the doublet
splitting to antisymmetric (Dzyaloshinskii--Moriya) exchange rather than
to the isotropic scalar-chirality term, and the ideal $SU(2)$ joint
eigenbasis and the exact conservation of all three charges of the present
model do not carry over to that anisotropic Hamiltonian; the comparison
fixes only the ratio of the energy scales. The quadruplet carries the
smallest population, so the transfers that raise $m$ at the lowest cost
are $C$, between the doublets, at $2\varepsilon$ per unit of charge, and
$B'$, from the excited doublet into the quadruplet, at
$(\Delta-\varepsilon)/2$. The former is cheaper exactly when
$\Delta>5\varepsilon$ (at $\Delta=5\varepsilon$ the two crossings
coincide and the two-level argument below does not apply as such). The
crossing analysis of
Sec.~\ref{sec:bound} then selects $\lambda_{\rm cr}=-2\varepsilon$, where
the levels $\kE{-}{-\tfrac12}$ and $\kE{+}{+\tfrac12}$ of
$H+\lambda\,\n\!\cdot\!\J$ coincide, and the passive state of the co-moving
Hamiltonian mixes them:
\begin{equation}
\begin{gathered}
\ket{c}=\cos\chi\,\kE{-}{-\tfrac12}+e^{i\eta}\sin\chi\,\kE{+}{+\tfrac12},\qquad
0\le\chi\le\tfrac{\pi}{2},\\[2pt]
\rho_*=\rho_G+(p_--p_+)\big(\ket c\bra c-\PE{-}{-\frac12}\big),
\end{gathered}
\label{supp:eq:cstate}
\end{equation}
the states written in the Letter. The pair has $\Delta m=1$. Although its
coherence connects states with $\Delta m=1$, it carries no transverse
total spin, because $\bra{+,\tfrac12}J_+\kE{-}{-\tfrac12}=0$: the
total-spin generators commute with $H$ and do not connect its distinct
chirality eigenspaces. Hence
\begin{equation}
\begin{gathered}
\langle\J\rangle_*=q\,\n,\qquad q=(p_--p_+)\sin^2\!\chi,\qquad 0\le q\le p_--p_+ ,\\[2pt]
E_{\min}(q)=E_G+2\varepsilon\,q,\qquad \lamcr=-2\varepsilon\,\n,
\end{gathered}
\label{supp:eq:strongex}
\end{equation}
exactly, by the same argument as in Secs.~\ref{sec:bound}
and~\ref{sec:minimizers}: the bound \eqref{supp:eq:weak} is attained at the
crossing and \eqref{supp:eq:cstate} saturates it. Beyond $q=c\equiv p_--p_+$,
where pair $C$ is fully exchanged and $\rho_*$ is again static, the
ordered level crossings continue $E_{\min}$ through three further
segments. With $d\equiv p_+-p_0>0$ and $(x)_+=\max(x,0)$,
\begin{equation}
\begin{aligned}
E_{\min}(q)&=E_G+2\varepsilon q+\tfrac{\Delta-5\varepsilon}{2}\,(q-c)_+
+\tfrac{\Delta-\varepsilon}{2}\,(q-c-2d)_+ +2\varepsilon\,(q-2c-2d)_+,\\
&\qquad\qquad 0\le q\le2c+3d,
\end{aligned}
\label{supp:eq:strongcont}
\end{equation}
with the slopes $2\varepsilon$, $(\Delta-\varepsilon)/2$,
$\Delta-\varepsilon$ and $\Delta+\varepsilon$ on segments of capacity $c$,
$2d$, $c$ and $d$: the successive crossings exchange $\kE{+}{-\tfrac12}$
with $\kE{\tfrac32}{\tfrac32}$, then $\kE{+}{+\tfrac12}$ with
$\kE{\tfrac32}{\tfrac32}$, and finally $\kE{-}{-\tfrac12}$ with
$\kE{\tfrac32}{\tfrac12}$. The cap $2c+3d=2p_-+p_+-3p_0$ is the orbit
bound of Corollary~\ref{cor:cap}. The capacities are not all small at low
temperature: for $\beta\to\infty$, $c\to\frac12$ and $d\to0$, but the
third segment keeps the capacity $\frac12$.

\paragraph{Motion and observables on the doublet-mixing branch.}
The only coherence of $\rho_*$ in the joint basis is
\begin{equation}
\begin{aligned}
\bra{-,-\tfrac12}\rho_*(t)\kE{+}{+\tfrac12}
&=(p_--p_+)\sin\chi\cos\chi\,e^{-i\eta}\,e^{\,i(E_+-E_-)t/\hbar}\\
&=\sqrt{q\,(p_--p_+-q)}\;e^{\,i(2\varepsilon t/\hbar-\eta)},
\end{aligned}
\label{supp:eq:cohex}
\end{equation}
the rotation about $\n$ at the angular velocity $|\lamcr|/\hbar=2\varepsilon/\hbar$
with the period $T=\pi\hbar/\varepsilon$; the amplitude vanishes at both
end points, where $\rho_*$ is diagonal and static. Because the coherence
has $\Delta m=1$, the local moments themselves move. Their components
along $\n$ follow from Lemma~\ref{lem:single} and the diagonal part of
$\rho_*$, $\langle\sigma_i^{\n}\rangle=\frac{2q}{3}$, as on the Gibbs
orbit. The transverse components are
\begin{equation}
\langle\sigma_i^+\rangle_*(t)
=\bra{-,-\tfrac12}\rho_*(t)\kE{+}{+\tfrac12}\;\bra{+,\tfrac12}\sigma_i^+\kE{-}{-\tfrac12},
\qquad
\big|\bra{+,\tfrac12}\sigma_i^+\kE{-}{-\tfrac12}\big|=\tfrac23 ,
\label{supp:eq:localmom}
\end{equation}
so that each site carries a transverse moment of length\linebreak
$|\langle\boldsymbol\sigma_i^{\perp}\rangle|=\allowbreak 2|\langle\sigma_i^+\rangle|
=\frac43\sqrt{q(p_--p_+-q)}$. The matrix element in \eqref{supp:eq:localmom}
is the same on all sites in magnitude, and its phase advances by
$2\pi/3$ from site to site because $\kE{-}{-\tfrac12}$ and $\kE{+}{+\tfrac12}$
carry opposite $C_3$ eigenvalues $\omega_3^{\pm1}$: the three transverse
moments are $120^\circ$ apart, form an umbrella around $\n$ with the
constant opening set by $q$, sum to zero, and rotate rigidly about $\n$
with the period $\pi\hbar/\varepsilon$. The total moment stays at $q\,\n$,
so a uniform probe sees no precession; a site-resolved one does. The
double-raising correlations $\langle\sigma_i^+\sigma_j^+\rangle$, which
need a $\Delta m=2$ coherence, vanish on this branch (other two-spin
correlations need not: at zero temperature the branch lies in the
total-spin-$\frac12$ sector, where
$\sum_{i<j}\langle\boldsymbol\sigma_i\!\cdot\!\boldsymbol\sigma_j\rangle=-3$).
The roles of the single-spin and the double-raising observables are thus
exchanged between the two branches, according to whether the coherence
that carries the time dependence has $\Delta m=1$ or $\Delta m=2$.
Table~\ref{tab:exchange} verifies all of this numerically.

\paragraph{Weak exchange, $0\le\Delta<\varepsilon$.}
For completeness we record the opposite regime, which is not used in the
Letter. The level order is that of $H$, $p_->p_0>p_+$, and $B$ is now
cheaper than $A$. The multiplier bound is maximized first at the $B$
crossing and then at the $A$ crossing:
\begin{equation}
E_{\min}(q)=E_G+\begin{cases}
\tfrac{\varepsilon-\Delta}{2}\,q, & 0\le q\le q_B=2(p_0-p_+),\\[3pt]
\tfrac{\varepsilon-\Delta}{2}\,q_B+\tfrac{\varepsilon+\Delta}{2}\,(q-q_B), & q_B\le q\le q_{\max}=2(p_--p_+),
\end{cases}
\label{supp:eq:weakex}
\end{equation}
with $|\lamcr|=(\varepsilon\mp\Delta)/2$ in the two segments. The
coherence lives in pair $B$ for $q<q_B$ and in pair $A$ for $q>q_B$, with
pair $B$ fully exchanged. After pair $B$ has been exchanged, the further
crossings at $\mu=\varepsilon-\Delta$ involve levels of equal population
and do not change the passive value; the next crossing with nonzero
capacity is $A$ at $\mu=(\varepsilon+\Delta)/2$, so \eqref{supp:eq:weakex}
holds throughout $0<\Delta<\varepsilon$. The double-raising correlations
oscillate at $(\varepsilon-\Delta)/\hbar$ and $(\varepsilon+\Delta)/\hbar$,
respectively: the frequency of the time crystal jumps at $q_B$, where
$E_{\min}$ has a kink; the point $q=q_B$ itself is static for $\Delta>0$.
At $\Delta=0$ the two crossings coincide and the full two-angle family of
Sec.~\ref{sec:minimizers} returns. The intermediate regime
$\varepsilon\le\Delta\le5\varepsilon$ follows from the same crossing
analysis but involves the $\Delta m=1$ pairs and is not reproduced here.

\section{Numerical method and results}
\label{sec:numerics}

\paragraph{Amplitudes and gauge.}
For finite positive temperature all eight populations of the Gibbs
spectrum are nonzero, and we choose the minimal full-rank amplitude
matrix, $W\equiv\widehat{\Psi}\in\mathbb C^{8\times8}$ with
$W(0)^\dagger W(0)=P=\mathrm{diag}(p_-,p_-,p_0,p_0,p_0,p_0,p_+,p_+)$; a
larger gauge dimension $r>8$ is allowed by the Letter but not needed. On
the $T=0$ orbit, whose spectrum $(\frac12,\frac12,0,\dots,0)$ has rank
two, the minimal choice is $W\in\mathbb C^{8\times2}$ with
$W(0)^\dagger W(0)=\mathrm{diag}(\frac12,\frac12)$. We use the gauge
$\mathcal A=0$ and the phase convention $\nu=0$, so that the
Uhlmann-amplitude form of Eq.~\eqref{eq:cov1} of the Letter reads
$i\hbar\dot W=(H+\hbar\kappa\,\calD_t\rho)W$.

\paragraph{The covariant velocity and the multipliers.}
At each evaluation the current density matrix $\rho=WW^\dagger$ is
diagonalized, $\rho=\sum_kp_k\ket{k}\bra{k}$, including its zero
eigenvalues when present. In this basis Eq.~\eqref{eq:cov} of the Letter is
algebraic,
\begin{equation}
(\calD_t\rho)_{kl}=\frac{(H+\lambda^aJ_a)_{kl}\,(p_l-p_k)}{\hbar\,[\,i-\kappa(p_l-p_k)\,]},
\label{supp:eq:Dalg}
\end{equation}
affine in $\lambda^a$. Charge conservation, $d\langle J_a\rangle/dt=
i\kappa\Tr(\calD_t\rho\,[J_a,\rho])=0$ [Eq.~\eqref{eq:charge} of the Letter], is then a real
$3\times3$ linear system $M_{ab}\lambda^b=-c_a$ with
$M_{ab}=i\Tr(\calD^{(b)}[J_a,\rho])$ and $c_a=i\Tr(\calD^{(0)}[J_a,\rho])$,
where $\calD^{(0)}$ and $\calD^{(b)}$ are \eqref{supp:eq:Dalg} evaluated with $H$
and with $J_b$, respectively. Setting $Q=c^aJ_a$ with real $c^a$ and
$d_{kl}=p_l-p_k$, Eq.~\eqref{supp:eq:Dalg} gives
\begin{equation}
c^aM_{ab}c^b=-\frac1\hbar\sum_{k,l}\frac{d_{kl}^2}{1+\kappa^2d_{kl}^2}\,|Q_{kl}|^2 ,
\label{supp:eq:Mdef}
\end{equation}
so the symmetric part of $M$ is negative definite exactly when the
commutators $[J_a,\rho]$ are linearly independent over the reals; this is
the uniqueness used in Sec.~\ref{sec:minimizers}. The system loses rank
whenever a real linear combination of the generators commutes with $\rho$:
at the Gibbs state, at the cap, and at the two static interior members of
\eqref{supp:eq:static}. A null direction $\mathbf c$ of $M$ has
$c^ac_a=i\Tr(\calD^{(0)}[\mathbf c\!\cdot\!\J,\rho])=0$, so an exact
compatible solution exists and the singularity is a matter of
conditioning. We solve the system by least squares with the minimum-norm
solution (singular values below $3\epsilon_{\rm mach}\approx7\times10^{-16}$
times the largest one are discarded) and monitor the residual
$\|M\boldsymbol\lambda+\mathbf c\|$ and the norm of
$\lambda^a[J_a,\rho]$, which is the combination that enters the flow.
Along every run of Tables~\ref{tab:attractors} and \ref{tab:series} the
residual stays below $2\times10^{-15}$, $\|\lambda^a[J_a,\rho]\|_{HS}$
below $0.35\,\varepsilon$ and $|\boldsymbol\lambda|$ below
$0.91\,\varepsilon$, including the approach to the singular Gibbs
endpoint. At the exact Gibbs state
$M=\mathbf c=0$ and the minimum-norm solution is $\boldsymbol\lambda=0$;
the nonzero limiting multiplier seen along a trajectory before that state
is reached is a path- and tolerance-dependent diagnostic
(Remark~\ref{rem:cusp}), not an invariant of the endpoint.

\paragraph{Interaction picture.}
Because $[H,J_a]=0$, the substitution $W=e^{-iHt/\hbar}W_I$ removes the
Hamiltonian rotation from the integration without approximation: the
eigenbasis of $\rho_I=W_IW_I^\dagger$ is that of $\rho$ rotated by
$e^{iHt/\hbar}$, the matrix elements of $H+\lambda^aJ_a$ in it are
unchanged, and so are $M_{ab}$ and $c_a$. Hence
$i\hbar\dot W_I=\hbar\kappa\,\calD_t\rho[\rho_I]\,W_I$ with the same
multipliers, and $\rho(t)=e^{-iHt/\hbar}\rho_I(t)e^{iHt/\hbar}$ exactly.
This separates the relaxation, which proceeds at rates of order
$10^{-2}$--$10^{-1}\,\omega$ (Table~\ref{tab:attractors}), from the
oscillation at $\sqrt3\,\omega$. The physical
velocity is reconstructed as $\|\dot\rho\|=\|[H+\hbar\kappa\calD_t\rho,\rho_I]\|/\hbar$,
and all observables are evaluated in the Schr\"odinger picture.

\paragraph{Integration, initial states and diagnostics.}
The amplitude equation is integrated with an adaptive eighth-order
Runge--Kutta method (DOP853) at relative tolerance $10^{-10}$ and absolute
tolerance $10^{-12}$, with \eqref{supp:eq:Dalg} and the $3\times3$ system solved
at every internal stage. Initial states are random points on the orbit:
a Haar-random unitary $U$ (pseudo-random generator seeded with 20261028,
the unitaries being drawn in the order of the table rows) is applied to
the diagonal Gibbs amplitudes, $W(0)=UP^{1/2}$ for the $8\times8$ case and
$W(0)=U\,\mathrm{diag}(\frac12,\frac12)^{1/2}$ restricted to two columns at
$T=0$, followed by a damped Newton correction along the orbit with the
spectrum-preserving generators $i[\rho,J_a]$ until
$|\langle\J\rangle-q\,\n|<10^{-13}$. No information about the minimizers
enters. Each run is integrated to the fixed final time $T_{\rm end}=600/\omega$
($3000/\omega$ at $\beta\varepsilon=4$, where the relaxation is about
seven times slower); the terminal covariant velocity $\|\calD_t\rho\|$,
which is the quantity that must vanish on a relative equilibrium, is
listed in the tables and lies between $10^{-11}$ and $10^{-6}$. All
quantities are evaluated at $t=T_{\rm end}$ and all norms are
Hilbert--Schmidt (Frobenius) norms. Decay rates are asymptotic rates: least-squares slopes of
$-\ln\|\calD_t\rho\|$ against $t$ over the last two decades above the
floor of about $10^{-10}$ set by the integration tolerance,
$10^{-8}\le\|\calD_t\rho\|\le10^{-6}$; this window excludes the initial
transient, and the slope is stable against its choice for every run with
$q>0$. Along every run the spectrum (hence the entropy) stays constant
to $4\times10^{-9}$ and $\langle\J\rangle$ to $2\times10^{-13}$, and the
energy decreases monotonically. The initial matrices, the trajectories and
the raw data of Tables~\ref{tab:attractors}--\ref{tab:exchange} and of
Fig.~\ref{fig:gibbs} of the Letter are deposited with the simulation codes in the
repository quoted in the Letter (tagged release and data-file names to be
given there).

\paragraph{Results.}
Tables~\ref{tab:attractors}--\ref{tab:exchange} list every run of the
Letter: the final states of Fig.~\ref{fig:gibbs}(a) reached from one random initial state
per charge on the Gibbs orbits at $\beta\varepsilon=\frac12,1,2,4$ and at
$T=0$ (Table~\ref{tab:attractors}); the two-runs-per-charge series of
Fig.~\ref{fig:gibbs}(b) at $\beta\varepsilon=1$, together with runs at the isolated
charges $q_A$ and $q_B$ of \eqref{supp:eq:static} (Table~\ref{tab:series}); and
the runs on the exchange-shifted orbit, $\Delta=6\varepsilon$,
$\beta\varepsilon=1$, of Fig.~\ref{fig:gibbs}(c) (Table~\ref{tab:exchange}). All numbers
are quoted at the end of the run.

On the Gibbs orbits every sampled run with $q>0$ in
Tables~\ref{tab:attractors} and \ref{tab:series} ends on the family
\eqref{supp:eq:rhostar}--\eqref{supp:eq:family}. Two diagnostics establish this:
the largest forbidden coherence, i.e.\ the largest matrix element in the
joint basis outside the diagonal and the two pair coherences (column
``forbidden coh.''), and the residual
$r_{\rm fam}=\min\|\rho-\rho_*(\theta,\vartheta,\phi,\varphi)\|_{HS}$ of
the fit to the family (column ``$r_{\rm fam}$''), which also tests the
diagonal; the latter is below $10^{-6}$ in every run with $q>0$, and below
$2\times10^{-8}$ at $\beta\varepsilon=1$ and at $T=0$. The angles read off the populations
of $\kE{\tfrac32}{\tfrac32}$ and $\kE{+}{+\tfrac12}$ satisfy the constraint
\eqref{supp:eq:qE} (column ``$\delta_q$''), the energy agrees with
\eqref{supp:eq:Emin} to $3\times10^{-10}\,\varepsilon$, the multiplier is
$-\frac{\varepsilon}{2}\n$ to $2\times10^{-11}$ in magnitude and to
$2\times10^{-13}$ in direction, the three local moments equal
$\frac{2q}{3}\n$ to $6\times10^{-9}$, $|\langle\sigma_i^\pm\rangle|$ is
below $4\times10^{-7}$, the amplitude of
$\langle\sigma_1^+\sigma_2^+\rangle$ lies inside the band of
\eqref{supp:eq:amp}, and the physical velocity $\|\dot\rho\|$ agrees with
Eq.~\eqref{supp:eq:rdot} to $10^{-9}$ (Table~\ref{tab:series}). Both azimuths
advance at the rate $-\varepsilon/\hbar$ with constant difference
$\phi-\varphi$. Runs with the same $q$ end at different
$(\sin^2\theta,\sin^2\vartheta)$ and different amplitudes, as the
degeneracy of the family implies, and the runs at the isolated charges
$q_A$ and $q_B$ end at moving members. The sampled runs at $q=0$ approach
$\rho_G$: at $T_{\rm end}$ their distance $\|\rho-\rho_G\|_{HS}$ is between
$7\times10^{-7}$ and $2\times10^{-6}$ and still decreasing, with a charge
error $|\langle\J\rangle|$ below $10^{-14}$, and the multiplier along these
particular trajectories has $|\boldsymbol\lambda|/\varepsilon=\frac12$ to
$3\times10^{-6}$ with a direction $\boldsymbol\lambda/|\boldsymbol\lambda|$
that differs from run to run (Remark~\ref{rem:cusp}). The covariant
velocity decays exponentially: for $q>0$ at the asymptotic rates
$0.020$--$0.021\,\omega$ at $\beta\varepsilon=\frac12$ and $2$,
$0.027$--$0.029\,\omega$ at $\beta\varepsilon=1$, $0.0038\,\omega$ at
$\beta\varepsilon=4$ and $0.20\,\omega$ at $T=0$, from early times on;
at $q=0$ a faster initial stage (about $0.05\,\omega$ at
$\beta\varepsilon=1$) precedes the final-stage rate listed in the tables. The physical velocity $\|\dot\rho\|$ stays finite for $q>0$
and vanishes only at $q=0$.

On the exchange-shifted orbit every sampled run with $q>0$ in
Table~\ref{tab:exchange} ends on the doublet-mixing branch
\eqref{supp:eq:cstate} (single-block fit residual $r_{\rm fam}<10^{-9}$): the
energy agrees with \eqref{supp:eq:strongex}, $|\lamcr|=2\varepsilon$ along
$-\n$, $\sin^2\chi=q/(p_--p_+)$, the coherence has the amplitude of
\eqref{supp:eq:cohex} and rotates at $2\varepsilon/\hbar$, and the local moments
precess as described after \eqref{supp:eq:localmom}: transverse length
$\frac43\sqrt{q(p_--p_+-q)}$ on every site, phase lags of $2\pi/3$, zero
transverse sum, constant component $\frac{2q}{3}$ along $\n$, and
vanishing double-raising correlations $\langle\sigma_i^+\sigma_j^+\rangle$.
The covariant velocity decays at the asymptotic rate $0.051\,\omega$ for $q>0$. The sampled run
at $q=0$ approaches $\rho_G$, to the distance $5\times10^{-6}$ at
$T_{\rm end}$, along a trajectory on which $|\boldsymbol\lambda|\to2\varepsilon$.

\clearpage
\setlength{\simulationtableheight}{\dimexpr\textwidth-12pt\relax}
\begin{sidewaystable}[p]
\centering
\begin{adjustbox}{max totalsize={\dimexpr\textheight-6pt\relax}{\simulationtableheight},center}
\begin{minipage}{\textheight}
\renewcommand{\baselinestretch}{1}\selectfont
\renewcommand{\small}{\fontsize{9}{11}\selectfont}
\renewcommand{\arraystretch}{0.88}
\centering
\caption{Final states of the covariant qLLG flow on Gibbs orbits [Fig.~\ref{fig:gibbs}(a) of
the Letter], one random initial state per row, $\kappa=0.5$,
$\n=\hat{\mathbf z}$, $q/q_{\max}=0,0.2,0.4,0.6,0.8,0.97$ with
$q_{\max}=\tanh(\beta\varepsilon/2)$ ($q_{\max}=1$ at $T=0$); every quantity
is evaluated at the end of the run, $t=T_{\rm end}$. $E$ is the energy and
$E-E_{\min}$ its deviation from Eq.~\eqref{supp:eq:Emin};
$\lambda_z/|\boldsymbol\lambda|$ is the direction of the multiplier ($-1$
means $\lamcr\parallel-\n$; at $q=0$ it depends on the initial state, see
Remark~\ref{rem:cusp}); $\sin^2\theta$ and $\sin^2\vartheta$ are read off
the populations of $\kE{\tfrac32}{\tfrac32}$ and $\kE{+}{+\tfrac12}$, and
$\delta_q=2(p_--p_0)\sin^2\theta+2(p_0-p_+)\sin^2\vartheta-q$ is the
residual of the constraint \eqref{supp:eq:qE}; ``forbidden coh.'' is the largest
matrix element of $\rho$ in the joint basis outside the diagonal and the
two pair coherences, a diagnostic of forbidden off-diagonal coherences;
$r_{\rm fam}=\|\rho-\rho_*(\theta,\vartheta,\phi,\varphi)\|_{HS}$ is the
residual of the fit to the family \eqref{supp:eq:rhostar}--\eqref{supp:eq:family},
with the angles taken from the populations and the phases from the pair
coherences (at $q=0$ the fitted member is $\rho_G$ itself and $r_{\rm fam}$
is the terminal distance from the Gibbs state); $\|\calD_t\rho\|$ is the
terminal covariant velocity; ``rate''
is the asymptotic decay rate of $\|\calD_t\rho\|$ in units of $\omega$,
the least-squares slope of $-\ln\|\calD_t\rho\|$ over the last two
decades above the integration floor, $10^{-8}\le\|\calD_t\rho\|\le10^{-6}$
($^{\dagger}$at $q=0$ this final stage is preceded by a faster initial
decay, see Sec.~\ref{sec:numerics}); the last two columns are the amplitude of
$\langle\sigma_1^+\sigma_2^+\rangle$ and the band allowed by
Eq.~\eqref{supp:eq:amp} at that charge, rounded outward to four decimals; at
$T=0$ the band is the single value $\sqrt{q(1-q)}/2\sqrt3$, given to five
decimals. All norms are Hilbert--Schmidt (Frobenius) norms and all times
are in units of $1/\omega$. Residuals, predictions and bands are evaluated
from unrounded values; the displayed coordinates are rounded independently.
At $T=0$ pair $B$ has no weight and $\sin^2\vartheta$ is undefined.}
\label{tab:attractors}
\footnotesize
\setlength{\tabcolsep}{3pt}
\begin{adjustbox}{max width=\linewidth,center}
\begin{tabular}{ccccccccccccccc}
\toprule
$\beta\varepsilon$ & $q$ & $E/\varepsilon$ & $(E-E_{\min})/\varepsilon$ & $|\lamcr|/\varepsilon$ & $\lambda_z/|\boldsymbol\lambda|$
& $\sin^2\theta$ & $\sin^2\vartheta$ & $\delta_q$ & forbidden coh. & $r_{\rm fam}$ & $\|\calD_t\rho\|$ & rate & $|\langle\sigma_1^+\sigma_2^+\rangle|$ & band \eqref{supp:eq:amp}\\
\midrule
$\tfrac12$ & 0.0000 & $-0.244919$ & $2\times10^{-11}$ & 0.500003 & $+0.4974$ & 0.0000 & 0.0000 & $10^{-11}$ & $3\times10^{-7}$ & $10^{-6}$ & $3\times10^{-9}$ & 0.0223$^{\dagger}$ & 0.0000 & [0.0000,\,0.0000] \\
$\tfrac12$ & 0.0490 & $-0.220427$ & $10^{-11}$ & 0.500000 & $-1.0000$ & 0.2786 & 0.0703 & $2\times10^{-12}$ & $3\times10^{-7}$ & $4\times10^{-7}$ & $4\times10^{-7}$ & 0.0202 & 0.0231 & [0.0000,\,0.0283] \\
$\tfrac12$ & 0.0980 & $-0.195935$ & $7\times10^{-12}$ & 0.500000 & $-1.0000$ & 0.5962 & 0.0766 & $-2\times10^{-11}$ & $5\times10^{-7}$ & $7\times10^{-7}$ & $6\times10^{-7}$ & 0.0200 & 0.0204 & [0.0049,\,0.0347] \\
$\tfrac12$ & 0.1470 & $-0.171443$ & $9\times10^{-12}$ & 0.500000 & $-1.0000$ & 0.4270 & 0.8851 & $3\times10^{-12}$ & $9\times10^{-8}$ & $2\times10^{-7}$ & $10^{-7}$ & 0.0200 & 0.0245 & [0.0049,\,0.0347] \\
$\tfrac12$ & 0.1959 & $-0.146951$ & $2\times10^{-11}$ & 0.500000 & $-1.0000$ & 0.8790 & 0.6697 & $9\times10^{-12}$ & $2\times10^{-7}$ & $5\times10^{-7}$ & $4\times10^{-7}$ & 0.0202 & 0.0258 & [0.0000,\,0.0283] \\
$\tfrac12$ & 0.2376 & $-0.126133$ & $2\times10^{-11}$ & 0.500000 & $-1.0000$ & 0.9930 & 0.9321 & $10^{-11}$ & $2\times10^{-7}$ & $3\times10^{-7}$ & $2\times10^{-7}$ & 0.0205 & 0.0090 & [0.0000,\,0.0121] \\
\addlinespace[3pt]
1 & 0.0000 & $-0.462117$ & $10^{-11}$ & 0.500000 & $+0.5078$ & 0.0000 & 0.0000 & $2\times10^{-12}$ & $4\times10^{-7}$ & $10^{-6}$ & $3\times10^{-11}$ & 0.0344$^{\dagger}$ & 0.0000 & [0.0000,\,0.0000] \\
1 & 0.0924 & $-0.415905$ & $2\times10^{-12}$ & 0.500000 & $-1.0000$ & 0.1237 & 0.4073 & $-4\times10^{-12}$ & $9\times10^{-10}$ & $2\times10^{-9}$ & $2\times10^{-9}$ & 0.0269 & 0.0292 & [0.0000,\,0.0534] \\
1 & 0.1848 & $-0.369694$ & $9\times10^{-12}$ & 0.500000 & $-1.0000$ & 0.3984 & 0.4043 & $-3\times10^{-12}$ & $2\times10^{-9}$ & $4\times10^{-9}$ & $4\times10^{-9}$ & 0.0270 & 0.0426 & [0.0277,\,0.0654] \\
1 & 0.2773 & $-0.323482$ & $6\times10^{-12}$ & 0.500000 & $-1.0000$ & 0.6108 & 0.5705 & $5\times10^{-12}$ & $7\times10^{-10}$ & $2\times10^{-9}$ & $10^{-9}$ & 0.0288 & 0.0634 & [0.0277,\,0.0654] \\
1 & 0.3697 & $-0.277270$ & $10^{-11}$ & 0.500000 & $-1.0000$ & 0.8693 & 0.6117 & $2\times10^{-12}$ & $3\times10^{-9}$ & $6\times10^{-9}$ & $5\times10^{-9}$ & 0.0270 & 0.0443 & [0.0000,\,0.0534] \\
1 & 0.4483 & $-0.237990$ & $2\times10^{-12}$ & 0.500000 & $-1.0000$ & 0.9996 & 0.8894 & $7\times10^{-14}$ & $4\times10^{-9}$ & $7\times10^{-9}$ & $6\times10^{-9}$ & 0.0270 & 0.0129 & [0.0000,\,0.0228] \\
\addlinespace[3pt]
2 & 0.0000 & $-0.761594$ & $4\times10^{-11}$ & 0.500003 & $+0.4862$ & 0.0000 & 0.0000 & $10^{-12}$ & $5\times10^{-7}$ & $2\times10^{-6}$ & $6\times10^{-9}$ & 0.0198$^{\dagger}$ & 0.0000 & [0.0000,\,0.0000] \\
2 & 0.1523 & $-0.685435$ & $10^{-11}$ & 0.500000 & $-1.0000$ & 0.1443 & 0.6112 & $10^{-12}$ & $10^{-7}$ & $2\times10^{-7}$ & $2\times10^{-7}$ & 0.0196 & 0.0782 & [0.0515,\,0.0880] \\
2 & 0.3046 & $-0.609275$ & $10^{-11}$ & 0.500000 & $-1.0000$ & 0.4028 & 0.3795 & $-6\times10^{-12}$ & $10^{-7}$ & $2\times10^{-7}$ & $2\times10^{-7}$ & 0.0200 & 0.1069 & [0.0807,\,0.1078] \\
2 & 0.4570 & $-0.533116$ & $-10^{-11}$ & 0.500000 & $-1.0000$ & 0.5903 & 0.6716 & $9\times10^{-12}$ & $2\times10^{-7}$ & $4\times10^{-7}$ & $3\times10^{-7}$ & 0.0199 & 0.0846 & [0.0807,\,0.1078] \\
2 & 0.6093 & $-0.456956$ & $2\times10^{-11}$ & 0.500000 & $-1.0000$ & 0.8523 & 0.4132 & $-10^{-11}$ & $5\times10^{-7}$ & $10^{-6}$ & $8\times10^{-7}$ & 0.0198 & 0.0803 & [0.0515,\,0.0880] \\
2 & 0.7387 & $-0.392221$ & $2\times10^{-11}$ & 0.500000 & $-1.0000$ & 1.0000 & 0.7483 & $6\times10^{-12}$ & $10^{-7}$ & $3\times10^{-7}$ & $2\times10^{-7}$ & 0.0197 & 0.0114 & [0.0000,\,0.0376] \\
\addlinespace[3pt]
4 & 0.0000 & $-0.964028$ & $7\times10^{-11}$ & 0.500000 & $-0.8222$ & 0.0000 & 0.0000 & $-10^{-11}$ & $6\times10^{-7}$ & $2\times10^{-6}$ & $10^{-9}$ & 0.0045$^{\dagger}$ & 0.0000 & [0.0000,\,0.0000] \\
4 & 0.1928 & $-0.867625$ & $5\times10^{-11}$ & 0.500000 & $-1.0000$ & 0.2007 & 0.1636 & $5\times10^{-11}$ & $7\times10^{-8}$ & $10^{-7}$ & $10^{-7}$ & 0.0038 & 0.1077 & [0.1049,\,0.1114] \\
4 & 0.3856 & $-0.771222$ & $2\times10^{-10}$ & 0.500000 & $-1.0000$ & 0.4048 & 0.1360 & $-10^{-10}$ & $3\times10^{-7}$ & $5\times10^{-7}$ & $4\times10^{-7}$ & 0.0038 & 0.1327 & [0.1312,\,0.1364] \\
4 & 0.5784 & $-0.674819$ & $5\times10^{-11}$ & 0.500000 & $-1.0000$ & 0.6109 & 0.0057 & $3\times10^{-11}$ & $10^{-7}$ & $2\times10^{-7}$ & $2\times10^{-7}$ & 0.0038 & 0.1329 & [0.1312,\,0.1364] \\
4 & 0.7712 & $-0.578417$ & $10^{-10}$ & 0.500000 & $-1.0000$ & 0.8141 & 0.0324 & $8\times10^{-11}$ & $10^{-7}$ & $2\times10^{-7}$ & $2\times10^{-7}$ & 0.0038 & 0.1064 & [0.1049,\,0.1114] \\
4 & 0.9351 & $-0.496474$ & $2\times10^{-11}$ & 0.500000 & $-1.0000$ & 0.9699 & 0.9778 & $4\times10^{-11}$ & $7\times10^{-8}$ & $2\times10^{-7}$ & $10^{-7}$ & 0.0038 & 0.0471 & [0.0297,\,0.0475] \\
\addlinespace[3pt]
$T=0$ & 0.0000 & $-1.000000$ & $2\times10^{-12}$ & 0.500000 & $-0.1479$ & 0.0000 & -- & $10^{-13}$ & $4\times10^{-7}$ & $10^{-6}$ & $5\times10^{-11}$ & 0.1934$^{\dagger}$ & 0.0000 & $\{0\}$ \\
$T=0$ & 0.2000 & $-0.900000$ & $-8\times10^{-13}$ & 0.500000 & $-1.0000$ & 0.2000 & -- & $-4\times10^{-13}$ & $10^{-11}$ & $10^{-9}$ & $6\times10^{-11}$ & 0.2038 & 0.1155 & $\{0.11547\}$ \\
$T=0$ & 0.4000 & $-0.800000$ & $-3\times10^{-12}$ & 0.500000 & $-1.0000$ & 0.4000 & -- & $2\times10^{-12}$ & $4\times10^{-11}$ & $10^{-9}$ & $2\times10^{-10}$ & 0.2038 & 0.1414 & $\{0.14142\}$ \\
$T=0$ & 0.6000 & $-0.700000$ & $-3\times10^{-12}$ & 0.500000 & $-1.0000$ & 0.6000 & -- & $-4\times10^{-13}$ & $10^{-11}$ & $10^{-9}$ & $10^{-10}$ & 0.2038 & 0.1414 & $\{0.14142\}$ \\
$T=0$ & 0.8000 & $-0.600000$ & $-10^{-12}$ & 0.500000 & $-1.0000$ & 0.8000 & -- & $-8\times10^{-15}$ & $2\times10^{-11}$ & $10^{-9}$ & $10^{-10}$ & 0.2038 & 0.1155 & $\{0.11547\}$ \\
$T=0$ & 0.9700 & $-0.515000$ & $5\times10^{-12}$ & 0.500000 & $-1.0000$ & 0.9700 & -- & $5\times10^{-13}$ & $2\times10^{-11}$ & $5\times10^{-10}$ & $8\times10^{-11}$ & 0.2038 & 0.0492 & $\{0.04924\}$ \\
\bottomrule
\end{tabular}
\end{adjustbox}
\end{minipage}
\end{adjustbox}
\end{sidewaystable}
\clearpage
\begin{sidewaystable}[p]
\centering
\begin{adjustbox}{max totalsize={\dimexpr\textheight-6pt\relax}{\simulationtableheight},center}
\begin{minipage}{\textheight}
\renewcommand{\baselinestretch}{1}\selectfont
\centering
\caption{The runs of Fig.~\ref{fig:gibbs}(b) of the Letter, $\beta\varepsilon=1$, two
random initial states for each of $q=0,0.2,0.4$, followed by one run at
each of the isolated charges $q_A=2(p_--p_0)$ and $q_B=2(p_0-p_+)$ of
Eq.~\eqref{supp:eq:static}; $T_{\rm end}=600/\omega$. ``rate'' is the decay rate
of $\|\calD_t\rho\|$ fitted as in Table~\ref{tab:attractors}
($^{\dagger}$preceded by a faster initial stage at $q=0$), and $\|\dot\rho\|$ the physical
velocity at the end of the run, followed in parentheses by the prediction
of Eq.~\eqref{supp:eq:rdot} evaluated from the unrounded angles; the other
columns are as in Table~\ref{tab:attractors}. The two runs with the same
$q$ end at different members of the family, and the runs at $q_A$ and
$q_B$ end at moving members, not at the static ones
$(\sin^2\theta,\sin^2\vartheta)=(1,0)$ and $(0,1)$.}
\label{tab:series}
\footnotesize
\setlength{\tabcolsep}{3.5pt}
\begin{adjustbox}{max width=\linewidth,center}
\begin{tabular}{cccccccccccc}
\toprule
$q$ & $(E-E_{\min})/\varepsilon$ & rate & $\|\dot\rho\|$ (pred.) & $|\lamcr|/\varepsilon$ & $\lambda_z/|\boldsymbol\lambda|$
& $\sin^2\theta$ & $\sin^2\vartheta$ & $r_{\rm fam}$ & $\|\calD_t\rho\|$ & $|\langle\sigma_1^+\sigma_2^+\rangle|$ & band \eqref{supp:eq:amp}\\
\midrule
0.0000 & $3\times10^{-11}$ & 0.0295$^{\dagger}$ & $3\times10^{-6}$ ($0$) & 0.500000 & $+0.1864$ & 0.0000 & 0.0000 & $2\times10^{-6}$ & $5\times10^{-11}$ & 0.0000 & [0.0000,\,0.0000] \\
0.0000 & $10^{-11}$ & 0.0278$^{\dagger}$ & $2\times10^{-6}$ ($0$) & 0.500000 & $+0.4576$ & 0.0000 & 0.0000 & $10^{-6}$ & $7\times10^{-11}$ & 0.0000 & [0.0000,\,0.0000] \\
0.2000 & $9\times10^{-12}$ & 0.0269 & 0.1798 (0.1798) & 0.500000 & $-1.0000$ & 0.2411 & 0.9539 & $8\times10^{-9}$ & $7\times10^{-9}$ & 0.0492 & [0.0294,\,0.0661] \\
0.2000 & $4\times10^{-12}$ & 0.0269 & 0.2030 (0.2030) & 0.500000 & $-1.0000$ & 0.3146 & 0.7541 & $7\times10^{-9}$ & $6\times10^{-9}$ & 0.0521 & [0.0294,\,0.0661] \\
0.4000 & $7\times10^{-12}$ & 0.0269 & 0.0972 (0.0972) & 0.500000 & $-1.0000$ & 0.9777 & 0.5609 & $10^{-8}$ & $10^{-8}$ & 0.0299 & [0.0000,\,0.0456] \\
0.4000 & $2\times10^{-11}$ & 0.0274 & 0.0766 (0.0766) & 0.500000 & $-1.0000$ & 0.9996 & 0.5013 & $8\times10^{-9}$ & $7\times10^{-9}$ & 0.0188 & [0.0000,\,0.0456] \\
$q_A=0.3378$ & $6\times10^{-12}$ & 0.0269 & 0.1977 (0.1977) & 0.500000 & $-1.0000$ & 0.7113 & 0.7848 & $5\times10^{-9}$ & $4\times10^{-9}$ & 0.0381 & [0.0000,\,0.0592] \\
$q_B=0.1243$ & $2\times10^{-12}$ & 0.0270 & 0.2000 (0.2000) & 0.500000 & $-1.0000$ & 0.3529 & 0.0406 & $10^{-8}$ & $10^{-8}$ & 0.0395 & [0.0000,\,0.0592] \\
\bottomrule
\end{tabular}
\end{adjustbox}
\end{minipage}
\end{adjustbox}
\end{sidewaystable}
\clearpage
\begin{sidewaystable}[p]
\centering
\begin{adjustbox}{max totalsize={\dimexpr\textheight-6pt\relax}{\simulationtableheight},center}
\begin{minipage}{\textheight}
\renewcommand{\baselinestretch}{1}\selectfont
\centering
\caption{Final states on the exchange-shifted orbit, $\Delta=6\varepsilon$,
$\beta\varepsilon=1$ [Fig.~\ref{fig:gibbs}(c) of the Letter, bottom], random initial
states, $T_{\rm end}=600/\omega$, $p_--p_+=0.3802$. $E_{\min}$ is
Eq.~\eqref{supp:eq:strongex}; $\sin^2\chi$ is read off the population of
$\kE{+}{+\tfrac12}$, ``coh'' is the modulus of the coherence
\eqref{supp:eq:cohex}, and $|\langle\sigma_i^+\rangle|$ the transverse amplitude
of the local moments (identical on the three sites to $10^{-9}$). The
predictions $\sin^2\chi=q/(p_--p_+)$, $|{\rm coh}|=\sqrt{q(p_--p_+-q)}$
and $|\langle\sigma_i^+\rangle|=\frac23\sqrt{q(p_--p_+-q)}$ agree with
the tabulated values of the runs with $q>0$ to better than $10^{-10}$, $7\times10^{-10}$ and $5\times10^{-10}$,
respectively; the run at $q=0$ ends at the Hilbert--Schmidt distance $5\times10^{-6}$
from the Gibbs state, with a charge error below $10^{-15}$ and a residual
coherence that is still decreasing. The
double-raising channel $\max_{i\neq j}|\langle\sigma_i^+\sigma_j^+\rangle|$
vanishes; ``forbidden coh.'' is the largest matrix element in the joint
basis outside the diagonal and the coherence \eqref{supp:eq:cohex};
$r_{\rm fam}$ is the residual of the single-block fit to
\eqref{supp:eq:cstate}; $\|\calD_t\rho\|$ and ``rate'' are as in
Table~\ref{tab:attractors} ($^{\dagger}$preceded by a faster initial stage at $q=0$). In every
run with $q>0$ the multiplier points along $-\n$ ($\lambda_z/|\boldsymbol\lambda|=-1$ to $<10^{-15}$), the phase of
$\langle\sigma_i^+\rangle$ advances from site to site by $2\pi/3$ to
$3\times10^{-14}$, the transverse sum $|\sum_i\langle\sigma_i^+\rangle|$ is below $8\times10^{-15}$, the
phase of $\langle\sigma_1^+\rangle$ advances at $3.4641\,\omega$ ($=2\varepsilon/\hbar$ to $6\times10^{-11}$),
and $\langle\sigma_i^{z}\rangle=2q/3$ on all sites to $2\times10^{-10}$ (at $q=0$, to $4\times10^{-6}$).}
\label{tab:exchange}
\footnotesize
\setlength{\tabcolsep}{4pt}
\begin{adjustbox}{max width=\linewidth,center}
\begin{tabular}{ccccccccccc}
\toprule
$q$ & $(E-E_{\min})/\varepsilon$ & $|\lamcr|/\varepsilon$ & $\sin^2\chi$ & coh & $|\langle\sigma_i^+\rangle|$
& $\max|\langle\sigma_i^+\sigma_j^+\rangle|$ & forbidden coh. & $r_{\rm fam}$ & $\|\calD_t\rho\|$ & rate\\
\midrule
0.00 & $10^{-10}$ & 2.000000 & 0.0000 & 0.0000 & 0.0000 & $3\times10^{-12}$ & $2\times10^{-6}$ & $4\times10^{-6}$ & $5\times10^{-11}$ & 0.0586$^{\dagger}$ \\
0.10 & $9\times10^{-11}$ & 2.000000 & 0.2630 & 0.1674 & 0.1116 & $4\times10^{-12}$ & $10^{-10}$ & $8\times10^{-10}$ & $4\times10^{-10}$ & 0.0511 \\
0.20 & $9\times10^{-11}$ & 2.000000 & 0.5261 & 0.1898 & 0.1266 & $4\times10^{-12}$ & $3\times10^{-11}$ & $10^{-9}$ & $2\times10^{-10}$ & 0.0511 \\
0.20 & $2\times10^{-11}$ & 2.000000 & 0.5261 & 0.1898 & 0.1266 & $2\times10^{-12}$ & $10^{-11}$ & $10^{-9}$ & $10^{-10}$ & 0.0511 \\
0.30 & $5\times10^{-11}$ & 2.000000 & 0.7891 & 0.1551 & 0.1034 & $4\times10^{-12}$ & $3\times10^{-11}$ & $8\times10^{-10}$ & $2\times10^{-10}$ & 0.0511 \\
\bottomrule
\end{tabular}
\end{adjustbox}
\end{minipage}
\end{adjustbox}
\end{sidewaystable}
\clearpage

\begingroup
\let\SavedSupplementLabel\label
\renewcommand{\label}[1]{%
  \def\SupplementLabelName{#1}%
  \def\RevtexLastBibLabel{LastBibItem}%
  \ifx\SupplementLabelName\RevtexLastBibLabel
    \SavedSupplementLabel{supp:LastBibItem}%
  \else
    \SavedSupplementLabel{#1}%
  \fi}

\endgroup


\begin{thebibliography}{99}
\urlstyle{same}

\bibitem{Wilczek2012} F.~Wilczek, Quantum time crystals, Phys. Rev. Lett. \textbf{109}, 160401 (2012), \url{https://doi.org/10.1103/PhysRevLett.109.160401}.

\bibitem{ShapereWilczek2012} A.~Shapere and F.~Wilczek, Classical time crystals, Phys. Rev. Lett. \textbf{109}, 160402 (2012), \url{https://doi.org/10.1103/PhysRevLett.109.160402}.

\bibitem{Else2016} D.~V. Else, B.~Bauer, and C.~Nayak, Floquet time crystals, Phys. Rev. Lett. \textbf{117}, 090402 (2016), \url{https://doi.org/10.1103/PhysRevLett.117.090402}.

\bibitem{Khemani2016} V.~Khemani, A.~Lazarides, R.~Moessner, and S.~L. Sondhi, Phase structure of driven quantum systems, Phys. Rev. Lett. \textbf{116}, 250401 (2016), \url{https://doi.org/10.1103/PhysRevLett.116.250401}.

\bibitem{Yao2017} N.~Y. Yao, A.~C. Potter, I.-D. Potirniche, and A.~Vishwanath, Discrete time crystals: Rigidity, criticality, and realizations, Phys. Rev. Lett. \textbf{118}, 269901 (2017), \url{https://doi.org/10.1103/PhysRevLett.118.269901}.

\bibitem{Zhang2017} J.~Zhang \emph{et al.}, Observation of a discrete time crystal, Nature \textbf{543}, 217 (2017), \url{https://doi.org/10.1038/nature21413}.

\bibitem{Choi2017} S.~Choi \emph{et al.}, Observation of discrete time-crystalline order in a disordered dipolar many-body system, Nature \textbf{543}, 221 (2017), \url{https://doi.org/10.1038/nature21426}.

\bibitem{Iemini2018} F.~Iemini, A.~Russomanno, J.~Keeling, M.~Schir\`o, M.~Dalmonte, and R.~Fazio, Boundary time crystals, Phys. Rev. Lett. \textbf{121}, 035301 (2018), \url{https://doi.org/10.1103/PhysRevLett.121.035301}.

\bibitem{Kongkhambut2022} P.~Kongkhambut, J.~Skulte, L.~Mathey, J.~G. Cosme, A.~Hemmerich, and H.~Ke\ss ler, Observation of a continuous time crystal, Science \textbf{377}, 670 (2022), \url{https://doi.org/10.1126/science.abo3382}.

\bibitem{Sacha2018} K.~Sacha and J.~Zakrzewski, Time crystals: A review, Rep. Prog. Phys. \textbf{81}, 016401 (2018), \url{https://doi.org/10.1088/1361-6633/aa8b38}.

\bibitem{Zaletel2023} M.~P. Zaletel, M.~Lukin, C.~Monroe, C.~Nayak, F.~Wilczek, and N.~Y. Yao, Colloquium: Quantum and classical discrete time crystals, Rev. Mod. Phys. \textbf{95}, 031001 (2023), \url{https://doi.org/10.1103/RevModPhys.95.031001}.

\bibitem{Bruno2013} P.~Bruno, Impossibility of spontaneously rotating time crystals: A no-go theorem, Phys. Rev. Lett. \textbf{111}, 070402 (2013), \url{https://doi.org/10.1103/PhysRevLett.111.070402}.

\bibitem{WatanabeOshikawa2015} H.~Watanabe and M.~Oshikawa, Absence of quantum time crystals, Phys. Rev. Lett. \textbf{114}, 251603 (2015), \url{https://doi.org/10.1103/PhysRevLett.114.251603}.

\bibitem{DNPW2019} J.~Dai, A.~J. Niemi, X.~Peng, and F.~Wilczek, Truncated dynamics, ring molecules, and mechanical time crystals, Phys. Rev. A \textbf{99}, 023425 (2019), \url{https://doi.org/10.1103/PhysRevA.99.023425}.

\bibitem{ADN2020} A.~Alekseev, J.~Dai, and A.~J. Niemi, Provenance of classical Hamiltonian time crystals, J. High Energy Phys. \textbf{08} (2020) 035, \url{https://doi.org/10.1007/JHEP08(2020)035}.

\bibitem{MarsdenRatiu} J.~E. Marsden and T.~S. Ratiu, \emph{Introduction to Mechanics and Symmetry}, 2nd ed. (Springer, New York, 1999), \url{https://doi.org/10.1007/978-0-387-21792-5}.

\bibitem{Coleman1985} S.~Coleman, Q-balls, Nucl. Phys. B \textbf{262}, 263 (1985), \url{https://doi.org/10.1016/0550-3213(86)90520-1}. 

\bibitem{LeePang1992} T.~D. Lee and Y.~Pang, Nontopological solitons, Phys. Rep. \textbf{221}, 251 (1992), \url{https://doi.org/10.1016/0370-1573(92)90064-7}.

\bibitem{GDN2021} J.~Garaud, J.~Dai, and A.~J. Niemi, Vortex precession and exchange in a Bose--Einstein condensate, J. High Energy Phys. \textbf{07} (2021) 157, \url{https://doi.org/10.1007/JHEP07(2021)157}.

\bibitem{Uhlmann1986} A.~Uhlmann, Parallel transport and ``quantum holonomy'' along density operators, Rep. Math. Phys. \textbf{24}, 229 (1986), \url{https://doi.org/10.1016/0034-4877(86)90055-8}.

\bibitem{Uhlmann1991} A.~Uhlmann, A gauge field governing parallel transport along mixed states, Lett. Math. Phys. \textbf{21}, 229 (1991), \url{https://doi.org/10.1007/BF00420373}.

\bibitem{Tronci2019} C.~Tronci, Momentum maps for mixed states in quantum and classical mechanics, J. Geom. Mech. \textbf{11}, 639 (2019), \url{https://doi.org/10.3934/jgm.2019032}.

\bibitem{Niemi2025} A.~J. Niemi, No massless Goldstone bosons in Hamiltonian time crystals, Phys. Rev. Lett. \textbf{135}, 230201 (2025), \url{https://doi.org/10.1103/sytq-kjbr}.

\bibitem{Liu2024} Y.~Liu, I.~P. Miranda, L.~Johnson, A.~Bergman, A.~Delin, D.~Thonig, M.~Pereiro, O.~Eriksson, V.~Azimi-Mousolou, and E.~Sj\"oqvist, Quantum analog of Landau--Lifshitz--Gilbert dynamics, Phys. Rev. Lett. \textbf{133}, 266704 (2024), \url{https://doi.org/10.1103/PhysRevLett.133.266704}.

\bibitem{LandauLifshitz1935} L.~D. Landau and E.~M. Lifshitz, On the theory of the dispersion of magnetic permeability in ferromagnetic bodies, Phys. Z. Sowjetunion \textbf{8}, 153 (1935); reprinted in \emph{Collected Papers of L.~D. Landau}, edited by D.~ter Haar (Pergamon, Oxford, 1965), p.~101, \url{https://doi.org/10.1016/B978-0-08-010586-4.50023-7}.

\bibitem{Gilbert2004} T.~L. Gilbert, A phenomenological theory of damping in ferromagnetic materials, IEEE Trans. Magn. \textbf{40}, 3443 (2004), \url{https://doi.org/10.1109/TMAG.2004.836740}.

\bibitem{Trif2008} M.~Trif, F.~Troiani, D.~Stepanenko, and D.~Loss, Spin-electric coupling in molecular magnets,
Phys. Rev. Lett. \textbf{101}, 217201 (2008), \url{https://doi.org/10.1103/PhysRevLett.101.217201}.

\bibitem{GeorgeotMila2010} B.~Georgeot and F.~Mila, Chirality of triangular antiferromagnetic clusters as a qubit, Phys. Rev. Lett. \textbf{104}, 200502 (2010), \url{https://doi.org/10.1103/PhysRevLett.104.200502}.

\bibitem{HJW1993} L.~P. Hughston, R.~Jozsa, and W.~K. Wootters, A complete classification of quantum ensembles having a given density matrix, Phys. Lett. A \textbf{183}, 14 (1993), \url{https://doi.org/10.1016/0375-9601(93)90880-9}.

\bibitem{Kibble1979} T.~W.~B. Kibble, Geometrization of quantum mechanics, Commun. Math. Phys. \textbf{65}, 189 (1979), \url{https://doi.org/10.1007/BF01225149}.

\bibitem{GKM2005} J.~Grabowski, M.~Ku\'s, and G.~Marmo, Geometry of quantum systems: Density states and entanglement, J. Phys. A \textbf{38}, 10217 (2005), \url{https://doi.org/10.1088/0305-4470/38/47/011}.

\bibitem{Gisin1981} N.~Gisin, A simple nonlinear dissipative quantum evolution equation, J. Phys. A \textbf{14}, 2259 (1981), \url{https://doi.org/10.1088/0305-4470/14/9/021}.

\bibitem{SenChitra1995} D.~Sen and R.~Chitra, Large-$U$ limit of a Hubbard model in a magnetic field: Chiral spin interactions and paramagnetism, Phys. Rev. B \textbf{51}, 1922 (1995), \url{https://doi.org/10.1103/PhysRevB.51.1922}.

\bibitem{BBMK2008} L.~N. Bulaevskii, C.~D. Batista, M.~V. Mostovoy, and D.~I. Khomskii, Electronic orbital currents and polarization in Mott insulators, Phys. Rev. B \textbf{78}, 024402 (2008), \url{https://doi.org/10.1103/PhysRevB.78.024402}.

\bibitem{Claassen2017} M.~Claassen, H.-C. Jiang, B.~Moritz, and T.~P. Devereaux, Dynamical time-reversal symmetry breaking and photo-induced chiral spin liquids in frustrated Mott insulators, Nat. Commun. \textbf{8}, 1192 (2017), \url{https://doi.org/10.1038/s41467-017-00876-y}.

\bibitem{SM} See Supplementary Material

\bibitem{Trif2010} M.~Trif, F.~Troiani, D.~Stepanenko, and D.~Loss, Spin electric effects in molecular antiferromagnets, Phys. Rev. B \textbf{82}, 045429 (2010), \url{https://doi.org/10.1103/PhysRevB.82.045429}.

\bibitem{Chiorescu2000} I.~Chiorescu, W.~Wernsdorfer, A.~M\"uller, H.~B\"ogge, and B.~Barbara, Butterfly hysteresis loop and dissipative spin reversal in the $S=1/2$, V$_{15}$ molecular complex, Phys. Rev. Lett. \textbf{84}, 3454 (2000), \url{https://doi.org/10.1103/PhysRevLett.84.3454}.

\bibitem{Boudalis2018} A.~K. Boudalis, J.~Robert, and P.~Turek, First demonstration of magnetoelectric coupling in a polynuclear molecular nanomagnet: Single-crystal EPR studies of [Fe$_3$O(O$_2$CPh)$_6$(py)$_3$]ClO$_4\cdot$py under static electric fields, Chem. Eur. J. \textbf{24}, 14896 (2018), \url{https://doi.org/10.1002/chem.201803038}.

\bibitem{HsiehHawrylak2010} C.-Y. Hsieh and P.~Hawrylak, Quantum circuits based on coded qubits encoded in chirality of electron spin complexes in triple quantum dots, Phys. Rev. B \textbf{82}, 205311 (2010), \url{https://doi.org/10.1103/PhysRevB.82.205311}.

\bibitem{Roushan2017} P.~Roushan \emph{et al.}, Chiral ground-state currents of interacting photons in a synthetic magnetic field, Nat. Phys. \textbf{13}, 146 (2017), \url{https://doi.org/10.1038/nphys3930}.

\bibitem{Lienhard2020} V.~Lienhard, P.~Scholl, S.~Weber, D.~Barredo, S.~de L\'es\'eleuc, R.~Bai, N.~Lang, M.~Fleischhauer, H.~P. B\"uchler, T.~Lahaye, and A.~Browaeys, Realization of a density-dependent Peierls phase in a synthetic, spin-orbit coupled Rydberg system, Phys. Rev. X \textbf{10}, 021031 (2020), \url{https://doi.org/10.1103/PhysRevX.10.021031}.



%
%
%
%
%

\end{thebibliography}

\begin{thebibliography}{9}
\small
\bibitem{vonNeumann1937} J. von Neumann, Some matrix-inequalities and
metrization of matric-space, Tomsk Univ. Rev. \textbf{1}, 286 (1937);
reprinted in \emph{Collected Works}, Vol.~IV (Pergamon, Oxford, 1962).
\bibitem{Birkhoff1946} G. Birkhoff, Tres observaciones sobre el algebra
lineal, Univ. Nac. Tucum\'an Rev. Ser.~A \textbf{5}, 147 (1946).
\bibitem{HLP} G. H. Hardy, J. E. Littlewood, and G. P\'olya,
\emph{Inequalities} (Cambridge University Press, Cambridge, 1934),
Sec.~10.2.
\bibitem{PuszWoronowicz1978} W. Pusz and S. L. Woronowicz, Passive states
and KMS states for general quantum systems, Commun. Math. Phys.
\textbf{58}, 273 (1978).
\bibitem{Lenard1978} A. Lenard, Thermodynamical proof of the Gibbs formula
for elementary quantum systems, J. Stat. Phys. \textbf{19}, 575 (1978).
\bibitem{Allahverdyan2004} A. E. Allahverdyan, R. Balian, and
Th. M. Nieuwenhuizen, Maximal work extraction from finite quantum systems,
Europhys. Lett. \textbf{67}, 565 (2004).
\bibitem{Iida2011} K. Iida, Y. Qiu, and T. J. Sato, Dzyaloshinsky--Moriya
interaction and long lifetime of the spin state in the Cu$_3$ triangular
spin cluster by inelastic neutron scattering measurements, Phys. Rev. B
\textbf{84}, 094449 (2011).
\end{thebibliography}
\end{document}